\documentclass[submission,copyright,creativecommons]{eptcs}
\providecommand{\event}{MSFP 2016} 
\usepackage{breakurl}              

\title{Eilenberg--Moore Monoids and \\Backtracking Monad Transformers}
\def\titlerunning{Eilenberg--Moore Monoids and Backtracking Monad Transformers}
\author{
Maciej Pir\'og
\institute{Department of Computer Science\\KU Leuven, Belgium}
\email{maciej.pirog@cs.kuleuven.be}
}
\def\authorrunning{M. Pir\'og}
\usepackage{amsmath}
\usepackage{amssymb}
\usepackage{mathrsfs}
\usepackage[only,llbracket,rrbracket]{stmaryrd}

\newcommand{\tuple}[1]{\langle #1 \rangle}
\newcommand{\bigp}[1]{\Bigl( #1 \Bigr)}
\newcommand{\catname}[1]{\mathbf{#1}}
\newcommand{\catvar}[1]{\mathscr{#1}}
\newcommand{\pholder}{(\text{--})}
\newcommand{\spholder}{\text{--}}
\newcommand{\emb}{\mathsf{emb}}
\newcommand{\cons}{\mathsf{cons}}
\newcommand{\ladj}[1]{\lfloor {#1} \rfloor}
\newcommand{\radj}[1]{\lceil {#1} \rceil}

\newcommand{\id}{\mathsf{id}}
\newcommand{\Id}{\mathsf{Id}}
\newcommand{\MID}{I}

\newcommand{\res}{{\mathcal R}}

\newcommand{\withm}[1]{#1 {\ltimes} M}
\newcommand{\foldfree}[1]{\llbracket {#1} \rrbracket}
\newcommand{\spanst}[1]{\widetilde{#1}}
\newcommand{\foldres}[1]{\langle\!\!\langle {#1} \rangle\!\!\rangle}

\newcommand{\INR}{\mathsf{INF}}

\newcommand{\LMT}{\mathsf{LMT}_M}

\newcommand{\monoidalexp}[2]{#1 \Rightarrow #2}
\newcommand{\mexp}{\monoidalexp{MA}{MA}}
\newcommand{\app}{\mathsf{app}}
\newcommand{\comp}{\mathsf{comp}}
\newcommand{\kcomp}{\mathsf{kcomp}}
\newcommand{\ident}{\mathsf{ident}}
\newcommand{\kident}{\mathsf{kident}}
\newcommand{\mladj}[1]{\lfloor {#1} \rfloor_{\otimes}}
\newcommand{\mradj}[1]{\lceil {#1} \rceil_{\otimes}}
\newcommand{\sym}{\mathsf{sym}}

\usepackage{amsthm}

\newtheorem{thm}{Theorem}
\newtheorem{lemma}[thm]{Lemma}
\theoremstyle{definition}
\newtheorem{defn}[thm]{Definition}
\newtheorem{remark}[thm]{Remark}
\newtheorem{example}[thm]{Example}

\usepackage{pifont}

\usepackage{tikz}
\tikzset
{
  baseline=(current  bounding  box.center),
  node distance=4em,
  f/.style = {fill=white,inner sep=0.1em},
}

\usepackage{fancyvrb}

\begin{document}

\maketitle

\begin{abstract}
We develop an algebraic underpinning of backtracking monad transformers in the general setting of monoidal categories.
As our main technical device, we introduce \emph{Eilenberg--Moore monoids}, which combine monoids with algebras for strong monads.
We show that Eilenberg--Moore monoids coincide with algebras for the list monad transformer (`done right') known from Haskell libraries.
From this, we obtain a number of results, including the facts that the list monad transformer is indeed a monad, a transformer, and an instance of the \texttt{MonadPlus} class.
Finally, we construct an Eilenberg--Moore monoid of endomorphisms, which, via the codensity monad construction, yields a continuation-based implementation \`a la Hinze.
\end{abstract}

\section{Introduction}

In monadic functional programming, the most straightforward approach to backtracking is realised by the list monad~\cite{DBLP:journals/jfp/Bird06,DBLP:conf/fpca/Wadler85}. More advanced structures are used for efficient implementation or for combining backtracking with other computational effects~\cite{DBLP:conf/icfp/Hinze00,DBLP:conf/icfp/KiselyovSFS05}. This paper is concerned with a category-theoretic explanation of such more advanced structures.

Historically, there has been some dispute over the `correct' definition of the monad transformer associated with the list monad. First, a popular Haskell library \texttt{mtl}\footnote{Documentation for each Haskell package is available online at \url{http://hackage.haskell.org/package/PACKAGE-NAME}.} proposed the following definition (in a pseudo-Haskell syntax):
\begin{Verbatim}
     type ListT m a = m [a]
\end{Verbatim}
The idea behind this type is that each computation first performs some effects in \texttt{m}, and then returns a list of results. This structure is not entirely satisfactory. One problem with it is strictly mathematical: the list monad transformer defined in this way is not really a monad transformer, since it does not satisfy the required equations in general (it does when the transformed monad \texttt{m} is commutative). A conceptual disadvantage is that it does not fully reflect the backtracking aspect of list computations: if we want to look only for the first result, we are not necessarily interested in the effects associated with the subsequent results of the computation, while in this implementation all \texttt{m}-effects are performed immediately.

In this paper, we consider an alternative definition of the list monad transformer, known as `\texttt{ListT} done right', which is provided by the Haskell packages \texttt{list-t}, \texttt{List}, and \texttt{pipes}. It can be implemented as follows:
\begin{Verbatim}
     data ListT m a = ListT (m (Maybe (a, ListT m a)))
\end{Verbatim}
The idea is that a value of this type is a list in which each tail is guarded by \texttt{m}. Intuitively, if we want to extract the $n$-th result from the computation, we first have to perform the effects guarding the elements $1$ to $n$.
Jaskelioff and Moggi~\cite{Jaskelioff20104441} propose an equational presentation of this monad, on which we expand and generalise in this paper (in particular, we abstract to general monoidal categories).

Another approach to backtracking transformers---taken, for example, by Hinze~\cite{DBLP:conf/icfp/Hinze00} and Kiselyov \textit{et al.}~\cite{DBLP:conf/icfp/KiselyovSFS05}---is to employ continuations. Hinze proposes the following monad, which uses a pair of success and failure continuations:
\begin{Verbatim}
     type Backtr m a = forall x. (a -> m x -> m x) -> m x -> m x
\end{Verbatim}
Hinze derives this monad using Hughes's~\cite{DBLP:conf/afp/Hughes95} `context-passing' technique. As one of our contributions, we relate it to the list monad transformer using the codensity monad construction~\cite{DBLP:conf/mpc/Hinze12}.

In this paper, we model structures like those mentioned above in a monoidal category~$\catvar C$ (or, when the continuation-based structures are involved, a closed monoidal category), starting with an algebraic specification of backtracking. As our point of departure, we define \emph{Eilenberg--Moore monoids}, which combine the theory of a strong monad (the `transformed' monad), the theory of monoids (which gives us choice and failure), and a coherence condition that specifies the order in which computations are executed. We construct the free monad induced by Eilenberg--Moore monoids, which coincides with the list monad transformer `done right', and introduce a `Cayley representation', which gives us an isomorphic continuation-based implementation.
In detail, our contributions are the following:

\begin{itemize}
\item In Section~\ref{sec:emmonoids}, given a strong monad $M$, we introduce \emph{Eilenberg--Moore $M$-monoids}. They are tuples $\tuple{A, a, m, u}$, where $\tuple{A,a}$ is an Eilenberg--Moore $M$-algebra, and $\tuple{A,m,u}$ is a monoid in the ambient monoidal category, that satisfy a coherence condition bringing together the algebra structure, the monoid structure, and the strength of the monad. Then, assuming certain algebraically free monads exist, we construct free Eilenberg--Moore monoids. The monad $UF$ induced by the free--underlying adjunction $F \dashv U$ is the list monad transformer applied to~$M$.

\item In Section~\ref{sec:beck}, we show that the category of Eilenberg--Moore $M$-monoids is isomorphic to the category of Eilenberg--Moore algebras for the list monad transformer applied to~$M$. As an application of this result, we employ a correspondence between monad morphisms and functors between Eilenberg--Moore categories to obtain a monad transformer structure.

\item In Section~\ref{sec:codensity}, assuming that the ambient category is closed, we introduce an Eilenberg--Moore monoid of endomorphisms $MA \Rightarrow MA$. We prove that it is a Cayley representation of a certain subcategory of Eilenberg--Moore monoids. Since this subcategory contains all free Eilenberg--Moore monoids, this gives us, via the codensity monad construction~\cite{DBLP:conf/mpc/Hinze12}, a~continuation-based monad transformer isomorphic to the list monad transformer. It turns out to coincide with the mentioned construction introduced by Hinze~\cite{DBLP:conf/icfp/Hinze00}.

\item In Section~\ref{sec:revisiting}, to show that the used techniques can be applied more universally, we revisit the transformer for commutative monads \verb+m [a]+. We characterise its algebras, which turn out to form a subcategory of Eilenberg--Moore $M$-monoids for a commutative monad~$M$. We introduce a Cayley representation $A \Rightarrow MA$ of a sufficient subcategory of these monoids, and obtain a novel continuation-based implementation.
\end{itemize}

Thus, we relate a number of existing constructions and introduce one new `commutative monad' transformer. We take a high-level approach---for example, we use the resumption monad and its universal properties to describe free Eilenberg--Moore monoids, which liberates us from tedious proofs by structural induction. The seemingly arbitrary structure related to the list monad transformer (like the \texttt{lift} monad morphism or the instance of the \texttt{MonadPlus} class) is simply a corollary of the general results, and is not only correct by construction, but also recognised as canonically related to the free--underlying adjunction $F \dashv U$.
The category-theoretic approach allows us to split the tackled constructions into more fine-grained pieces; for example, the role of monadic strength becomes apparent, although in Haskell, where all monads are canonically strong, it is usually used implicitly, obfuscating equational reasoning.

\section{Background}

We denote categories by $\catvar C$, $\catvar D$, \ldots. We work in a monoidal category $\tuple{\catvar C, \otimes, I}$ (we additionally assume that it is closed in Section~\ref{sec:codensity}, and symmetric closed in Section~\ref{sec:revisiting}). We use the symbol $\cong$ for the structural natural isomorphisms of monoidal categories (instead of the more traditional $\lambda$, $\rho$, and $\alpha$). We always give types explicitly, so no confusion should arise. Also, we skip the subscripts in natural transformations when the component is obvious from the type (so, for example, we write $MMA \xrightarrow{\mu} MA$ instead of $MMA \xrightarrow{\mu_A} MA$).

\subsection{Monoids in monoidal categories}

A \emph{monoid} in $\catvar{C}$ is a triple $\tuple{A,\ A\otimes A \xrightarrow{m} A,\ I \xrightarrow{u} A}$, where $m$ and $u$ are called the \emph{multiplication} and the \emph{unit} respectively, such that the following diagrams commute:
\begin{equation*}
\begin{tikzpicture}
\node(a0){$I \otimes A$};
\node(b0)[right of=a0, xshift=2em]{$A \otimes A$};
\node(c0)[right of=b0, xshift=2em]{$A \otimes I$};
\node(b1)[below of=b0]{$A$};
\path[arrows={-latex}, font=\scriptsize]
(a0) edge node [auto] {$u \otimes \id$} (b0)
(a0) edge node [left, xshift=-0.7em] {$\cong$} (b1)
(b0) edge node [auto] {$m$} (b1)
(c0) edge node [right, xshift=0.7em] {$\cong$} (b1)
(c0) edge node [above] {$\id \otimes u$} (b0)
;
\end{tikzpicture}
\hspace{2em}
\begin{tikzpicture}
\node(a0){$(A \otimes A) \otimes A$};
\node(a1)[below of=a0]{$A \otimes (A \otimes A)$};
\node(b1)[right of=a1,xshift=4em]{$A \otimes A$};
\node(c1)[right of=b1,xshift=4em]{$A$};
\node(c0)[above of=c1]{$A\otimes A$};
\path[arrows={-latex}, font=\scriptsize]
(a0) edge node [auto] {$m \otimes \id$} (c0)
(a0) edge node [auto] {$\cong$} (a1)
(a1) edge node [auto] {$\id \otimes m$} (b1)
(b1) edge node [auto] {$m$} (c1)
(c0) edge node [auto] {$m$} (c1)
;
\end{tikzpicture}
\end{equation*}
A morphism between monoids $\tuple{A,m^A,u^A}$ and $\tuple{B,m^B,u^B}$ is a morphism $h : A \to B$ in $\catvar C$ such that $h \cdot m^A = m^B \cdot (h \otimes h)$ and $h \cdot u^A = u^B$. The category of monoids and morphisms between monoids is called $\catname{Mon}$. If the obvious forgetful functor $U_{\catname{Mon}} : \catname{Mon} \to \catvar C$ has a left adjoint $F_{\catname{Mon}}$, the induced monad is a generalisation of the list monad.

\subsection{Monads and strength}

To set the notation, we give a number of basic definitions of monads and some related concepts. A \emph{monad} on $\catvar C$ is a triple $\tuple{M,\ MM \xrightarrow{\mu} M ,\ \Id \xrightarrow{\eta} M}$ such that it is a monoid in the monoidal category of endofunctors on $\catvar C$ with natural transformations as morphisms, and the monoidal tensor given by composition of endofunctors. We denote a monad simply by its carrier $M$, and the monadic structure, that is, the multiplication and the unit, are always denoted $\mu$ and $\eta$ respectively. If there are a number of monads in the context, we sometimes put the name of the monad in superscript, for example $\mu^M$ and $\eta^M$. The category of monads on $\catvar C$ and monad morphisms (that is, morphism between appropriate monoids) is denoted $\catname{Mnd}$.

An \emph{Eilenberg--Moore algebra} (or simply an \emph{algebra}) for a monad $M$ on $\catvar C$ is a pair $\tuple{A,\ MA \xrightarrow{a} A}$, where $A$ is an object in $\catvar C$, while $a$ is a morphism such that $a \cdot Ma = a \cdot \mu_A$ and $a \cdot \eta_A = \id_A$. If $a : MA \to A$ is a morphism such that the pair $\tuple{A,a}$ is an Eilenberg--Moore algebra, we say that $a$ has the Eilenberg--Moore property. A morphism between two Eilenberg--Moore algebras $\tuple{A,a}$ and $\tuple{B,b}$ is a morphism $h : A \to B$ such that $b \cdot Mh = h \cdot a$. The category of such algebras and such morphisms is called the \emph{Eilenberg--Moore category} of~$M$.

For a monad $M$ on $\catvar C$ there is an adjunction between $\catvar C$ and the Eilenberg--Moore category of $M$ given on objects as $A \mapsto \tuple{MA, \mu_A}$ in one direction and $\tuple{A,a} \mapsto A$ in the other direction. For every other adjunction $L \dashv R : \catvar C \to \catvar D$ that induces $M$, there exists a functor (called the \emph{comparison functor}) from $\catvar D$ to the Eilenberg--Moore category of~$M$. It is defined on objects as
\begin{equation*}
X \mapsto \tuple{RX,\ MRX \xrightarrow{=} RLRX \xrightarrow{R\epsilon} RX},
\end{equation*}
where $\epsilon$ is the counit of the adjunction $L \dashv R$. If the comparison functor is an isomorphism, we say that the adjunction $L \dashv R$ is \emph{strictly monadic}. An example of a strictly monadic adjunction is $F_{\catname{Mon}} \dashv U_{\catname{Mon}}$. This entails that the Eilenberg--Moore category of the list monad is isomorphic to the category of monoids.

An endofunctor $G$ is \emph{strong} if it is equipped with a transformation $\tau : GA \otimes B \to G(A \otimes B)$ (called a \emph{strength} of $G$) natural in $A$ and $B$ such that the following diagrams commute:
\begin{equation*}
\begin{tikzpicture}
\node(a0){$GA \otimes I$};
\node(b0)[right of=a0, xshift=2em]{$G(A \otimes I)$};
\node(b1)[below of=b0]{$GA$};
\path[arrows={-latex}, font=\scriptsize]
(a0) edge node [auto] {$\tau$} (b0)
(a0) edge node [auto,xshift=-0.5em] {$\cong$} (b1)
(b0) edge node [auto] {$\cong$} (b1)
;
\end{tikzpicture}
\hspace{1em}
\begin{tikzpicture}
\node(a0){$GA \otimes (B \otimes C)$};
\node(a1)[below of=a0]{$(GA \otimes B) \otimes C$};
\node(b1)[right of=a1,xshift=6em]{$G(A \otimes B) \otimes C$};
\node(c1)[right of=b1,xshift=6em]{$G((A \otimes B) \otimes C)$};
\node(c0)[above of=c1]{$G(A \otimes (B\otimes C))$};
\path[arrows={-latex}, font=\scriptsize]
(a0) edge node [auto] {$\tau$} (c0)
(a0) edge node [auto] {$\cong$} (a1)
(a1) edge node [auto] {$\tau \otimes \id$} (b1)
(b1) edge node [auto] {$\tau$} (c1)
(c0) edge node [auto] {$\cong$} (c1)
;
\end{tikzpicture}
\end{equation*}
A monad $M$ is \emph{strong} if it is strong as an endofunctor, and additionally the following diagrams commute:
\begin{equation*}
\begin{tikzpicture}
\node(a0){$A \otimes B$};
\node(b0)[right of=a0, xshift=2em]{$MA \otimes B$};
\node(b1)[below of=b0]{$M(A \otimes B)$};
\path[arrows={-latex}, font=\scriptsize]
(a0) edge node [auto] {$\eta \otimes \id$} (b0)
(a0) edge node [auto,xshift=-0.5em] {$\eta$} (b1)
(b0) edge node [auto] {$\tau$} (b1)
;
\end{tikzpicture}
\hspace{1em}
\begin{tikzpicture}
\node(a0){$MMA \otimes B$};
\node(a1)[below of=a0]{$M(MA \otimes B)$};
\node(b1)[right of=a1,xshift=6em]{$MM(A\otimes B)$};
\node(c1)[right of=b1,xshift=6em]{$M(A \otimes B)$};
\node(c0)[above of=c1]{$MA \otimes B$};
\path[arrows={-latex}, font=\scriptsize]
(a0) edge node [auto] {$\mu \otimes \id$} (c0)
(a0) edge node [auto] {$\tau$} (a1)
(a1) edge node [auto] {$M\tau$} (b1)
(b1) edge node [auto] {$\mu$} (c1)
(c0) edge node [auto] {$\tau$} (c1)
;
\end{tikzpicture}
\end{equation*}
While a monad can have more than one strength, in $\catname{Set}$ (and in Haskell), every monad is equipped with a canonical strength given by $\tuple{m,b} \mapsto (M(\lambda a.\, \tuple{a,b}))(m)$. A strong monad on $\catvar C$ can be turned into a monoid in $\catvar C$:

\begin{thm}[Wolff~\cite{wolff}]
\label{thm:wolff}
Let $M$ be a strong monad on $\catvar C$. Then, the tuple
$\tuple{M\MID ,\, \mathfrak m ,\, \mathfrak u}$,
where
\begin{equation*}
\mathfrak m = \bigp{
M\MID \otimes M\MID \xrightarrow{\tau} M(\MID\otimes M\MID) \xrightarrow{\cong} MM\MID \xrightarrow{\mu} M\MID
}
\text{\ \ and\ \ }
\mathfrak u = \bigp{
\MID \xrightarrow{\eta} M\MID
},
\end{equation*}
is a monoid.
\end{thm}

\subsection{Strongly generated algebraically free monads}

Given an endofunctor $G$ on $\catvar C$, consider the category of $G$-algebras, that is, pairs $\tuple{A,\, GA \xrightarrow{a} A}$, and morphisms $\tuple{A,a} \to \tuple{B,b}$ given by $\catvar C$-morphisms $h : A \to B$ such that $b \cdot Gh = h \cdot a$. If the obvious forgetful functor from the category of $G$-algebras to $\catvar C$ has a left adjoint, we denote the free algebra generated by an object $A$ as $\tuple{G^*A,\, GG^*A \xrightarrow{\cons} G^*A}$, and call the induced monad $G^*$ the \emph{algebraically free} monad generated by $G$. This adjunction is strictly monadic, which entails that the category of $G$-algebras is the Eilenberg--Moore category of $G^*$. We denote the monadic structure of the algebraically free monad as $\mu^{\mathcal F}$ and $\eta^{\mathcal F}$.
The freeness property of $\tuple{G^*A,\cons}$ can be described as follows, where the morphism $\emb$ is given by $GA \xrightarrow{G\eta^{\mathcal F}} GG^*A \xrightarrow{\cons} G^*A$. 

\begin{thm}
For a morphism $GA \xrightarrow{g} A$, there exists a unique morphism $G^*A \xrightarrow{\foldfree{g}} A$ such that the following diagram commutes:
\begin{equation*}
\begin{tikzpicture}
\node(a0){$GA$};
\node(b0)[right of=a0, xshift=2em]{$G^*A$};
\node(b1)[below of=b0]{$A$};
\path[arrows={-latex}, font=\scriptsize]
(a0) edge node [auto] {$\emb$} (b0)
(a0) edge node [auto,xshift=-0.5em] {$g$} (b1)
(b0) edge node [auto] {$\foldfree g$} (b1)
;  
\end{tikzpicture}
\end{equation*}
\end{thm}

If $\catvar C$ has coproducts, the carrier of the free algebra $\tuple{G^*A, \cons}$ can be given as the carrier of the initial $(G\pholder + A)$-algebra (if it exists), that is, $G^*A = \mu X. GX + A$.
For endofunctors $G$ and $H$, if both $G^*$ and $H^*$ exist, every natural transformation $h : G \to H$ induces a monad morphism $h^* : G^* \to H^*$ given as $G^*A \xrightarrow{G^*\eta^{\mathcal F}} G^*H^*A \xrightarrow{\foldfree{\cons \cdot h}} H^*A$. Note that we use the notation $\pholder^*$ as if it was a functor, although in general not all endofunctors induce algebraically free monads.
If $G$ is strong, we define in what way $G^*$ can inherit the strength of $G$ in a coherent fashion:

\begin{defn}
Let $G$ be a strong endofunctor such that $G^*$ exists. We say that $G^*$ is \emph{strongly generated} if for all morphisms $f : A \otimes B \to C$ and a $G$-algebra $\tuple{C,c}$, there exists a unique morphism $\widehat f$ that makes the following diagram commute:
\begin{equation*}
\begin{tikzpicture}
\node(a0){$GG^*A \otimes B$};
\node(a1)[below of=a0]{$G^*A \otimes B$};
\node(a2)[below of=a1]{$A \otimes B$};
\node(b0)[right of=a0, xshift=4em]{$G(G^*A \otimes B)$};
\node(c0)[right of=b0, xshift=4em]{$GC$};
\node(c1)[below of=c0]{$C$};
\path[arrows={-latex}, font=\scriptsize]
(a0) edge node [auto] {$\tau$} (b0)
(a0) edge node [auto] {$\cons \otimes \id$} (a1)
(a2) edge node [right] {$\eta^{\mathcal F} \otimes \id$} (a1)
(a1) edge node [auto] {$\widehat f$} (c1)
(a2) edge node [below] {$f$} (c1)
(b0) edge node [auto] {$G\widehat f$} (c0)
(c0) edge node [auto] {$c$} (c1)
;
\end{tikzpicture}
\end{equation*}
\end{defn}

A strongly generated monad is strong with the strength given by $\widehat \eta$ for $\eta:A\otimes B \to M(A\otimes B)$.
Moreover, for a natural transformation $h : G \to H$, the monad morphism $h^* : G^* \to H^*$ preserves strength, that is, $h^* \cdot \tau = \tau \cdot (h^* \otimes \id)$.
If the category $\catvar C$ is closed (for example, $\catname{Set}$), all algebraically free monads generated by strong endofunctors are strongly generated (see Fiore~\cite[Theorem~4.4]{fiorestrong}). In this paper, we assume that all the algebraically free monads that we deal with are strongly generated.

\subsection{The resumption monad}

Another construction that we use is the \emph{resumption monad} introduced by Moggi~\cite{citeulike:763238}, also known as the \emph{free monad transformer}. Given a monad $M$ and an endofunctor $G$, it is given as the composition $M(GM)^*$ if $(GM)^*$ exists. In the case of the free monad given by initial algebras, it becomes $A \mapsto M (\mu X.GMX + A)$. Using the rolling lemma~\cite{DBLP:conf/ctcs/BackhouseBGW95}, it is isomorphic to $A \mapsto \mu X.M(GX + A)$.

We notice that the endofunctor part of the list monad transformer that we work with in this paper can be given by initial algebras $A \mapsto \mu X . M((A \otimes X) + I) \cong M(A \otimes M\pholder)^*I$, which are similar in shape to the resumption monad for the endofunctor $(A\otimes\spholder)$ applied to the object $I$ (and this is how it is implemented in the Haskell package \texttt{pipes}). Although the monadic structure of $\mathsf{LMT}_M$ is not given directly by the monadic structure of the resumption monad, the two are related.

Hyland, Plotkin, and Power~\cite{DBLP:journals/tcs/HylandPP06} show two important properties of the resumption monad. The first one is that it is induced by a distributive law $\lambda : (GM)^*M \to M(GM)^*$. Thus, its monadic structure can be defined as $\mu^{\mathcal R} = \mu^M \mu^{\mathcal F} \cdot M\lambda(GM)^*$ and $\eta^{\mathcal R} = \eta^M \eta^{\mathcal F}$. Moreover, a natural transformation $h : G \to H$ induces a monad morphism $M(hM)^* : M(GM)^* \to M(HM)^*$. In general, the composition of two strong endofunctors is strong, and the composition of two strong monads via a distributive law is strong. Thus, if $M$ and $G$ are strong, and $(GM)^*$ is strongly generated, the resumption monad $M(GM)^*$ is strong.

The other important property of the resumption monad is that $M(GM)^*$ is the coproduct in $\catname{Mnd}$ of $M$ and $G^*$.
A classical result by Kelly~\cite{BAZ:4759448} states that the coproduct in $\catname{Mnd}$ of two monads $M$ and $T$ is always given by the free--underlying adjunction between the base category and the category of tuples $\tuple{A,\ MA\xrightarrow{m} A,\ TA \xrightarrow{t} A}$, where both $m$ and $t$ have the Eilenberg--Moore property. Moreover, this adjunction is strictly monadic.
Thus, using the correspondence between Eilenberg--Moore algebras for algebraically free monads and algebras for endofunctors, Hyland, Plotkin, and Power describe the Eilenberg--Moore category of the resumption monad for an endofunctor $G$ and a monad $M$ as consisting of tuples $\tuple{A,\ MA\xrightarrow{a} A,\ GA \xrightarrow{g} A}$ such that $a$ has the Eilenberg--Moore property.
Adapting Kelly's result, the freeness property of the resumption monad can be stated as follows, where the natural transformation $\INR$ (injection of the functor) is defined as:
\begin{equation*}
\INR = \Bigl(
G \xrightarrow{G\eta^M} GM \xrightarrow{\emb} (GM)^* \xrightarrow{\eta^M} M(GM)^*
\Bigr)
\end{equation*}

\begin{thm}
\label{thm:foldresuniprop}
Given an Eilenberg--Moore algebra $\tuple{B,\ MB \xrightarrow{b} B}$, an algebra $\tuple{B,\ GB\xrightarrow{g}B}$, and a morphism $h : A \to B$, there exists a unique morphism $\foldres{b,g,h} : M(GM)^*A \to B$ such that the following diagrams commute:
\begin{equation*}
\begin{tikzpicture}
\node(a1){$M(GM)^*A$};
\node(a0)[left of=a1,xshift=-3em]{$MM(GM)^*A$};
\node(a2)[right of=a1,xshift=4em]{$M(GM)^*M(GM)^*A$};
\node(a3)[right of=a2,xshift=5em]{$GM(GM)^*A$};
\node(b1)[below of=a1]{$B$};
\node(b0)[below of=a0]{$MB$};
\node(b3)[below of=a3]{$GB$};
\path[arrows={-latex}, font=\scriptsize]
(a1) edge node [auto] {$\foldres{b,g,h}$} (b1)
(a0) edge node [auto] {$\mu^M$} (a1)
(b0) edge node [auto] {$b$} (b1)
(a0) edge node [auto] {$M\foldres{b,g,h}$} (b0)
(a3) edge node [above, xshift=0.2em] {$\INR$} (a2)
(a2) edge node [above, xshift=0.3em] {$\mu^\res$} (a1)
(a3) edge node [auto] {$G\foldres{b,g,h}$} (b3)
(b3) edge node [above] {$g$} (b1)
;
\end{tikzpicture}
\hspace{1em}
\begin{tikzpicture}
\node(a1){$M(GM)^*A$};
\node(a0)[left of=a1,xshift=-2em]{$A$};
\node(b1)[below of=a1]{$B$};
\path[arrows={-latex}, font=\scriptsize]
(a0) edge node [auto] {$\eta^\res$} (a1)
(a1) edge node [auto] {$\foldres{b,g,h}$} (b1)
(a0) edge node [left, xshift=-0.5em] {$h$} (b1)
;
\end{tikzpicture}
\end{equation*}
\end{thm}

Morphisms $\foldres{\spholder,\spholder,\spholder}$ enjoy some equational properties that we find useful in the remainder of this paper:

\begin{lemma}
\label{thm:foldreseqprops}
Given $b$, $g$, and $h$ as in the previous theorem, the following hold:
\begin{enumerate}
\item The morphism $\foldres{b,g,h}$ can be expressed as $M(GM)^*A \xrightarrow{M(GM)^*h} M(GM)^*B \xrightarrow{\foldres{b,g,\id}} B$.
\item The morphism $\foldres{b,g,\id}$ has the Eilenberg--Moore property, and it is equal to $M(GM)^*B \xrightarrow{M\foldfree{g \cdot Gb}} MB \xrightarrow{b} B$.
\item The morphism $GA \xrightarrow{\INR} M(GM)^*A \xrightarrow{\foldres{b,g,h}} B$ is equal to $GA \xrightarrow{Gh} GB \xrightarrow{g} B$. 
\end{enumerate}
\end{lemma}

\section{Eilenberg--Moore monoids}
\label{sec:emmonoids}

In this section, we introduce Eilenberg--Moore monoids, which serve as an algebraic specification of backtracking combined with other effects. Then, we construct free Eilenberg--Moore monoids, which induce the list monad transformer.

\begin{defn}
\label{defn:emmonoids}
Let $M$ be a strong monad on a monoidal category $\catvar C$. An \emph{Eilenberg--Moore $M$-monoid} is a tuple
\begin{equation*}
\tuple{A,\ \ MA \xrightarrow{a} A ,\ \ A \otimes A \xrightarrow{m} A,\ \ \MID \xrightarrow{u} A}
\end{equation*}
such that the following hold:
\begin{enumerate}
\item $\tuple{A, a}$ is an Eilenberg--Moore $M$-algebra,
\item $\tuple{A, m, u}$ is a monoid,
\item coherence: the following diagram commutes:
\begin{equation*}
\begin{tikzpicture}
\node(a) {$MA \otimes A$};
\node(b)[below of=a] {$M(A \otimes A)$};
\node(c)[right of=a, xshift=10em] {$A \otimes A$};
\node(d)[right of=b, xshift=4em] {$MA$};
\node(e)[right of=d, xshift=2em] {$A$};
\path[arrows={-latex}, font=\scriptsize]
(a) edge node [auto] {$a \otimes \id$} (c)
(a) edge node [auto] {$\tau$} (b)
(b) edge node [auto] {$Mm$} (d)
(d) edge node [auto] {$a$} (e)
(c) edge node [auto] {$m$} (e)
;
\end{tikzpicture}
\end{equation*}
\end{enumerate}
A morphism between two Eilenberg--Moore $M$-monoids is a morphism in $\catvar C$ that is both a morphism between $M$-algebras and between monoids. We call the category of Eilenberg--Moore $M$-monoids and such morphisms $\catname{EMMon}_M$.
\end{defn}

\begin{example}
\label{ex:muladd}
Let $\catvar C$ be $\catname{Set}$. Let $M = G^*$ be the free monad generated by the functor $GX = X \times X$, that is, $G^*$ is the free monad of the theory of a single binary operation. Consider the tuple $\tuple{\mathbb N,\ \foldfree{(+)} : G^*\mathbb N \to \mathbb N ,\, (*) : \mathbb N \times \mathbb N \to \mathbb N ,\ \lambda x. 1 : \MID \to \mathbb N}$, in which we interpret the monad operation as addition, while the monoid is given by natural numbers with multiplication. It is an Eilenberg--Moore monoid. In this case, the coherence condition amounts to the right-distributivity of multiplication over addition.
\end{example}

Although to the author's best knowledge Eilenberg--Moore $M$-monoid is a new concept, it is an obvious generalisation of the concept of $F$-monoid:

\begin{defn}
[Fiore, Plotkin, and Turi~\cite{DBLP:conf/lics/FiorePT99}]
\emph{$F$-monoids} are similar to Eilenberg--Moore monoids, but we drop the condition (1) from Definition~\ref{defn:emmonoids} and the assumption that the endofunctor is a monad (that is, it is merely a strong endofunctor).
\end{defn}

We also need the following technical lemma, which relates $F$-monoids to Eilenberg--Moore monoids for free monads. Note that it could be used for a simple proof that the tuple from Example~\ref{ex:muladd} is an Eilenberg-Moore monoid.

\begin{lemma}
\label{thm:emfreemonoidfromfmonoid}
Let $G$ be a strong endofunctor that strongly generates a free monad $G^*$. Then, each $G$-monoid $\tuple{A ,\  GA \xrightarrow{g} A ,\  A\otimes A \xrightarrow{m} A ,\  I \xrightarrow{u} A}$ gives rise to an Eilenberg--Moore $G^*$-monoid $\tuple{A ,\  G^*A \xrightarrow{\foldfree{g}} A,\ m,\ u}$.
\end{lemma}

From the perspective of algebraic effects~(see Hyland and Power~\cite{DBLP:journals/tcs/HylandP06}), Eilenberg--Moore $M$-monoids can be seen as the right-distributive tensor of the theory of monoids over the theory of the monad~$M$. Each Eilenberg--Moore monoid consists of a monoid and an interpretation of operations provided by~$M$, while, denoting the monoid multiplication as $\vee$, for an $n$-ary $M$-operation~$f$, the coherence diagram becomes the following equation:
\begin{equation*}
f(x_1 , \ldots , x_n) \vee y = f(x_1 \vee y, \ldots , x_n \vee y)
\end{equation*}
Intuitively, it states that when making a choice between two values, the effects in the left-hand argument are always executed first. This differentiates backtracking from plain nondeterminism, where no order of execution is imposed. Note that the list monad is used for backtracking usually in lazy languages, where laziness is the effect that defines the order in which elements arrive. In eager languages, backtracking can be implemented using \emph{lazy lists}, in which laziness is an explicit effect. Indeed, the type of lazy lists is an instance of the list monad transformer (see Section~\ref{sec:coind}).

\subsection{Free Eilenberg--Moore monoids}

Now, we describe free Eilenberg--Moore monoids, which give us the monadic structure of the list monad transformer. First, we need an auxiliary definition:

\begin{defn}
For a monad $M$ and a $\catvar C$-object $A$, we define the functor $(\withm A)X = A \otimes M X$.
\end{defn}

Note that each morphism $f : A \to B$ induces a natural transformation $\withm f : \withm A \to \withm B$. If $M$ is strong, the functor $\withm A$ is strong via $(A \otimes MB) \otimes C \xrightarrow{\cong} A \otimes (MB \otimes C) \xrightarrow{\id \otimes \tau} A \otimes M(B \otimes C)$. If the monad $(\withm A)^*$ exists and is strongly generated, we denote the resulting strength as $\widetilde \tau : (\withm A)^*B \otimes C \to (\withm A)^*(B \otimes C)$.
Our main result follows.

\begin{thm}
\label{thm:adj}
Let $M$ be a strong monad. Assume that there exists a strongly generated algebraically free monad $(\withm A)^*$ for all objects $A$. Then, the obvious forgetful functor $U : \catname{EMMon}_M \to \catvar C$ has a left adjoint $F : \catvar C \to \catname{EMMon}_M$ given as follows:
\begin{align*}
& F A = \tuple{
M(\withm A)^* \MID,\ \mu^M,\ \mathfrak m,\ \mathfrak u
}
\\
& F (f : A \to B) = M(\withm f)^* \MID,
\end{align*}
where $\tuple{M(\withm A)^* \MID,\, \mathfrak m,\, \mathfrak u}$ is the monoid induced by the resumption monad as described in Theorem~\ref{thm:wolff}.
\end{thm}

In detail, the associated natural isomorphism
$\ladj\spholder : \catname{EMMon}_M(FA, \tuple{B,b,m^B,u^B}) \cong \catvar C(A, B) : \radj\spholder$
is defined as follows.
Let $f : FA \to \tuple{B,b,m^B,u^B}$ be a morphism between Eilenberg--Moore $M$-monoids. The $\catvar C$-morphism $\ladj f : A \to B$ is defined as:
\begin{equation*}
A \xrightarrow{\cong} A \otimes I \xrightarrow{\INR\,} M(\withm A)^*I \xrightarrow{f} B
\end{equation*}
In the other direction, let $g : A \to U\tuple{B,b,m^B,u^B}$ be a $\catvar C$-morphism. Then, the $\catname{EMMon}_M$-morphism $\radj g : FA \to \tuple{B,b,m^B,u^B}$ is defined as:
\begin{equation*}
M(\withm A)^*I \xrightarrow{\foldres{
MB \xrightarrow{b} B ,\ 
A \otimes B \xrightarrow{g \otimes \id} B \otimes B \xrightarrow{m^B} B ,\  
I \xrightarrow{u^B} B
}} B
\end{equation*}

\begin{defn}
We call the monad $UF$ induced by the adjunction above the \emph{list monad transformer} and denote it as $\LMT$.
\end{defn}

To get a more direct definition of the monadic structure of $\LMT$, let $\epsilon_{\tuple{A,a,m^A,u^A}} = \radj{A \xrightarrow{\id} A}$ be the counit of the adjunction. The monad multiplication is thus given as follows:
\begin{equation*}
\Bigl(
UFUF \xrightarrow{U\epsilon F} UF
\Bigr)
=
\Bigl(
M(\withm{M(\withm A)^*})^*I \xrightarrow{\radj{\id} = \foldres{\mu^M ,\, \mathfrak m ,\, \mathfrak u}} M(\withm A)^*I
\Bigr)
\end{equation*}
The unit of $\LMT$ is given as:
$\ladj{\id} = \Bigl(
A \xrightarrow{\cong} A \otimes I \xrightarrow{\INR} M(\withm A)^*I
\Bigr)$.

We can also verify that the morphisms $\mathfrak m$ and $\mathfrak u$ form a \verb+MonadPlus+ structure~\cite{DBLP:conf/ppdp/RivasJS15}. They obviously form a monoid, so it is left to verify two additional laws: left distributivity and left zero (or, in the language of Plotkin and Power~\cite{DBLP:journals/acs/PlotkinP03}, that $\mathfrak m$ and $\mathfrak u$ are \emph{algebraic}). The desired laws are simply the preservation of the multiplication and the unit by $U\epsilon F$, which follows from the fact that $\epsilon F$ is a morphism between Eilenberg--Moore monoids.

\section{Algebras for the list monad transformer}
\label{sec:beck}

The previous section shows a construction of an adjunction $F \dashv U$ that gives rise to the monad $\LMT$, but it is not an ordinary adjunction: we show that it is strictly monadic. We use this fact to construct some monad morphisms.

\begin{thm}
\label{thm:monadicity}
The adjunction $F \dashv U$ is strictly monadic. This entails that the category $\catname{EMMon}_M$ is isomorphic to the category of Eilenberg--Moore algebras of the monad $\LMT$.
\end{thm}

As an application of Theorem~\ref{thm:monadicity}, we show that $\mathsf{LMT}$ is indeed a monad transformer, that is, we construct a monad morphism from a monad $M$ to $\LMT$. Instead of defining the morphism directly and mundanely verifying the necessary properties, we utilise the following theorem (see Barr and Wells~\cite[Ch. 3, Theorem 6.3]{ttt}):

\begin{thm}
Let $\catname{EM}$ be a category in which objects are Eilenberg--Moore categories of monads on~$\catvar C$, while morphisms are carrier-preserving functors (that is, functors that commute with the forgetful functors from Eilenberg--Moore categories to $\catvar C$). Then, there exists an isomorphism $\catname{Mnd} \cong \catname{EM}^{\mathsf{op}}$, where $\catname{Mnd}$ is the category of monads  on~$\catvar C$ and monad morphisms.
\end{thm}

In detail, given two monads $T$ and $M$, consider a carrier-preserving functor $F$, and let $\tuple{TA,\ MTA \xrightarrow{a} TA} = F\tuple{TA, \mu^T}$. Now, the monad morphism corresponding to $F$ is given for an object $A$ as $MA \xrightarrow{M\eta^T} MTA \xrightarrow{Ma} TA$.

There exist obvious forgetful functors from $\catname{EMMon}_M$ to the category of monoids and to the category of algebras for~$M$. These give us two monad morphisms: from the list monad and from $M$ respectively. The latter is the desired \texttt{lift} operation of monad transformers. Following the description above, it is given as follows:
\begin{equation*}
\Bigl(
MA \xrightarrow{M\eta^{\LMT}} MM(\withm A)^*I \xrightarrow{\mu^M} M(\withm A)^*I
\Bigr)
=
\Bigl(
MA \xrightarrow{M\cong} M(A \otimes I) \xrightarrow{M\emb} M(\withm A)^*I
\Bigr)
\end{equation*}

Moreover, there exists a forgetful functor from $\catname{EMMon}_M$ to the Eilenberg--Moore category of the resumption monad generated by the endofunctor $A \mapsto A \otimes A$. This forgetful functor is given as $\tuple{A,a,m,u} \mapsto \tuple{A,a,m}$, and it induces a monad morphism from the resumption monad to $\mathsf{LMT}_M$, which flattens the tree structure into a list.

\section{Continuation-based implementation}
\label{sec:codensity}

In this section, we deal with a continuation-based backtracking monad transformer \`a la Hinze~\cite{DBLP:conf/icfp/Hinze00} mentioned in the introduction. We derive it from the list monad transformer using the codensity monad construction, thus automatically obtaining that the two monads are isomorphic. First, we discuss some background on closed monoidal categories, which we need to model continuations.

\subsection{Background: closed monoidal categories}

A monoidal category $\tuple{\catvar C, {\otimes}, I}$ is \emph{closed} if for all $\catvar C$-objects $B$, the functor $\pholder \otimes B$ has a right adjoint $B \Rightarrow \pholder$. The associated natural isomorphisms $\mladj\spholder : \catvar C(A \otimes B, C) \cong \catvar C(A, B\Rightarrow C) : \mradj\spholder$ are currying and uncurrying respectively. We call the counit of this adjunction $\app : (A \Rightarrow B) \otimes A \to B$. Note that for a morphism $g : A \otimes B \to C$, it is the case that the morphism $A\otimes B \xrightarrow{\mladj g \otimes \id} (B\Rightarrow C) \otimes B \xrightarrow{\app} C$ is equal to~$g$. 

For objects $A$, $B$, and $C$, we define the following morphism $k$:
\begin{equation*}
k = \Bigl(
((B \Rightarrow C) \otimes (A \Rightarrow B)) \otimes A
\xrightarrow{\cong}
(B \Rightarrow C) \otimes ((A \Rightarrow B) \otimes A)
\xrightarrow{\id \otimes \app}
(B \Rightarrow C) \otimes B
\xrightarrow{\app}
C
\Bigr)
\end{equation*}
We define the `composition' morphism $\comp = \mladj k : (B \Rightarrow C) \otimes (A \Rightarrow B) \to A \Rightarrow C$ and the identity morphism $\ident = \mladj{I \otimes A \xrightarrow{\cong} A} : I \to (A \Rightarrow A)$. The triple $\tuple{A\Rightarrow A ,\, \comp ,\, \ident}$ forms a monoid.

\subsection{Cayley representation of `Kleisli' Eilenberg--Moore monoids}

The codensity monad of a functor $G : \catvar D \to \catvar E$ is given by the right Kan extension of $G$ along itself $\mathsf{Ran}_G G : \catvar E \to \catvar E$ (see Mac Lane~\cite[Ch. X]{lane1998categories} or Leinster~\cite{leinstercod}). Using the coend representation of Kan extensions, one can implement the codensity monad of a Haskell functor \texttt{f} as follows:
\begin{Verbatim}
     data Cod f a = Cod (forall x. (a -> f x) -> f x)
\end{Verbatim}
It is known that the codensity monad of a right adjoint is isomorphic to the monad induced by the adjunction. Hinze~\cite{DBLP:conf/mpc/Hinze12} gives the following example of how one can use this fact to derive a continuation-based implementation of the list monad. First, we can simulate the forgetful functor from the category of monoids using a class constraint:
\begin{Verbatim}
     data L1 a = L1 (forall w. (Monoid w) => (a -> w) -> w)
\end{Verbatim}
Now, instead of relying on instances of the \texttt{Monoid} class, one can use the universal monoid of endomorphisms \verb+x -> x+. A classic result by Cayley states that every monoid can be represented as a submonoid of the universal monoid (this submonoid is called the \emph{Cayley representation}). Thus, we can equivalently define the monad in question as follows:
\begin{Verbatim}
     data L2 a = L2 (forall x. (a -> x -> x) -> x -> x)
\end{Verbatim}
In this section, we give a similar construction to obtain a continuation-based implementation of the list monad transformer. We start with an Eilenberg--Moore monoid of endomorphisms $MA \Rightarrow MA$:

\begin{thm}
\label{thm:mama}
Let $\catvar C$ be closed monoidal. Then, the tuple
\begin{equation*}
\tuple{
\mexp
,\ 
\mladj{p}
,\ 
\comp
,\ 
\ident
},
\end{equation*}
where
\begin{equation*}
p = \Bigl(
M(\mexp) \otimes MA
\xrightarrow{\tau}
M((\mexp) \otimes MA)
\xrightarrow{M\app}
MMA
\xrightarrow{\mu}
MA
\Bigr),
\end{equation*}
is an Eilenberg--Moore monoid.
\end{thm}

Unfortunately, the Eilenberg--Moore monoid defined in Theorem~\ref{thm:mama} is not universal. For that, we would have to define a morphism $\tuple{A,a,m,u} \to \tuple{\mexp,\mladj{p},\comp,\ident}$ for each Eilenberg--Moore monoid $\tuple{A,a,m,u}$, while it is in general not possible to define a morphism $g : A \to (MA \Rightarrow MA)$ that is a morphism between Eilenberg--Moore algebras $\tuple{A,a}$ and $\tuple{MA\Rightarrow MA, \mladj{p}}$. To make the construction work, we need to slightly restrict the domain of the forgetful functor $U$.

First, we give some intuition. The universal property of the adjunction $F \dashv U$ described in Theorem~\ref{thm:adj} is a folding property: given morphisms $a : MA \to A$, $m : A \otimes A \to A$, $u : I \to A$, and $h : B \to A$, as long as $a$, $m$, and $u$ satisfy the conditions given in the definition of Eilenberg--Moore monoids, we obtain a unique coherent \textit{fold}, that is, a morphism $\mathsf{LMT}_M B \to A$. It could be also understood as `running' or `interpreting' the monadic computation. However, in programming, when we `run' a backtracking computation, we do not interpret it as a value of some type $A$. Rather, we interpret it as a value in the base monad, that is, $MA$. In other words, we fold the structure of the list, but, instead of eliminating the monadic parts using an Eilenberg--Moore algebra $a : MA \to A$, we accumulate it using the monadic multiplication $\mu$.

Thus, we are interested in Eilenberg--Moore monoids of the shape $\tuple{MA, \mu, m, u}$, which we call, for the sake of this article, \emph{Kleisli monoids}, referring to the known fact that the full subcategory of the Eilenberg--Moore category of a monad $M$ that consists of algebras of the shape $\tuple{MA, \mu}$ is equivalent to the Kleisli category of~$M$. We call the full subcategory of $\catname{EMMon}_M$ that consists of Kleisli monoids $\catname{KlMon}_M$. The restriction of the forgetful functor $U : \catname{EMMon}_M \to \catvar C$ to $\catname{KlMon}_M$ is dubbed $U_{\catname{Kl}} : \catname{KlMon}_M \to \catvar C$. 

A useful observation is that free Eilenberg--Moore monoids defined in Theorem~\ref{thm:adj} are also Kleisli monoids. This means that $U_{\catname{Kl}}$ has a left adjoint $F_{\catname{Kl}}$ defined in the same way as~$F$, and that the monad induced by $F_{\catname{Kl}} \dashv U_{\catname{Kl}}$ is the same monad as the one induced by $F \dashv U$, that is, $\mathsf{LMT}_M$. Therefore, the monad $\mathsf{LMT}_M$ is also isomorphic to the codensity monad of $U_{\catname{Kl}}$. This way, it is enough for our purposes to find a Cayley representation of Kleisli monoids, not necessarily all Eilenberg--Moore monoids. The Eilenberg--Moore monoid $MA \Rightarrow MA$ from Theorem~\ref{thm:mama}, although not a Kleisli monoid itself, is universal for Kleisli monoids:

\begin{thm}
\label{thm:rep}
For each Kleisli monoid $\tuple{MA, \mu, m, u}$, the morphism $\mladj{m} : MA \to (MA \Rightarrow MA)$ has the following properties:
\begin{itemize}
\item it is an Eilenberg--Moore monoid morphism $\tuple{MA, \mu, m, u} \to \tuple{\mexp,\,\mladj{p},\,\comp,\,\ident}$,
\item it is a split monomorphism in $\catvar C$, that is, there exists a morphism $r : (MA \Rightarrow MA) \to MA$ in $\catvar C$ such that $r \cdot \mladj{m} = \id$.
\end{itemize}
\end{thm}

Using the codensity monad for this representation yields the following monad transformer:
\begin{Verbatim}
     type Backtr m a = forall x. (a -> m x -> m x) -> m x -> m x 
\end{Verbatim}
It is the same monad transformer as obtained, although using different methods, by Hinze~\cite{DBLP:conf/icfp/Hinze00}.

\section{Revisiting the `effects-first' transformer for commutative monads}
\label{sec:revisiting}

Now, we revisit the commutative-monad transformer \verb+m [a]+ known from the \texttt{mtl} library in Haskell. We call it a `commutative-monad transformer', as it is a monad if and only if the transformed monad is commutative (see, for example, Mulry~\cite{EPTCS1296}). In this section, we derive its continuation-based isomorph, recreating the steps for the $\mathsf{LMT}_M$ monad presented in previous sections. We assume that $\catvar C$ is symmetric closed, and that $M$ is commutative.

\subsection{Background: symmetric monoidal categories and commutative monads}

A monoidal category is \emph{symmetric} if it is equipped with a natural isomorphism $A \otimes B \xrightarrow{s} B \otimes A$ such that $\sym$ is an involution (that is, $\sym \cdot \sym = \id$) and the following diagrams commute:

\begin{equation*}
\begin{tikzpicture}
\node(a0){$I \otimes A$};
\node(b0)[right of=a0, xshift=2em]{$A \otimes I$};
\node(b1)[below of=b0]{$A$};
\path[arrows={-latex}, font=\scriptsize]
(a0) edge node [auto] {$\sym$} (b0)
(a0) edge node [left, xshift=-0.7em] {$\cong$} (b1)
(b0) edge node [auto] {$\cong$} (b1)
;
\end{tikzpicture}
\hspace{2em}
\begin{tikzpicture}
\node(a0){$(A \otimes B) \otimes C$};
\node(a1)[below of=a0]{$A \otimes (B \otimes C)$};
\node(b1)[right of=a1,xshift=5em]{$(B \otimes C) \otimes A$};
\node(b0)[above of=b1]{$(B \otimes A) \otimes C$};
\node(c1)[right of=b1,xshift=5em]{$B \otimes (C \otimes A)$};
\node(c0)[above of=c1]{$B \otimes (A \otimes C)$};
\path[arrows={-latex}, font=\scriptsize]
(a0) edge node [auto] {$\sym \otimes \id$} (b0)
(b0) edge node [auto] {$\cong$} (c0)
(a0) edge node [auto] {$\cong$} (a1)
(a1) edge node [auto] {$\sym$} (b1)
(b1) edge node [auto] {$\cong$} (c1)
(c0) edge node [auto] {$\id \otimes \sym$} (c1)
;
\end{tikzpicture}
\end{equation*}

In a symmetric monoidal category, we define a \emph{left strength} for a strong monad $M$:
\begin{equation*}
\tau' = \Bigl(
A \otimes MB \xrightarrow{\sym} MB \otimes A \xrightarrow{\tau} M(B\otimes A) \xrightarrow{M\sym} M(A \otimes B)
\Bigr)
\end{equation*}
One can show that the appropriate mirror images of the diagrams for a strong endofunctor and a strong monad commute for $\tau'$. A monad is \emph{commutative} if it is equipped both with a (right) strength and a left strength, and the following diagram commutes for all objects $A$ and $B$:
\begin{equation*}
\begin{tikzpicture}
\node(a0){$MA \otimes MB$};
\node(a1)[below of=a0]{$M(A \otimes MB)$};
\node(b1)[right of=a1,xshift=5em]{$MM(A \otimes B)$};
\node(b0)[above of=b1]{$M(MA \otimes B)$};
\node(c1)[right of=b1,xshift=5em]{$M(A \otimes B)$};
\node(c0)[above of=c1]{$MM(A \otimes B)$};
\path[arrows={-latex}, font=\scriptsize]
(a0) edge node [auto] {$\tau'$} (b0)
(b0) edge node [auto] {$M\tau$} (c0)
(a0) edge node [auto] {$\tau$} (a1)
(a1) edge node [auto] {$M\tau'$} (b1)
(b1) edge node [auto] {$\mu$} (c1)
(c0) edge node [auto] {$\mu$} (c1)
;
\end{tikzpicture}
\end{equation*}

If $\catvar C$ is a closed symmetric monoidal category and $M$ is commutative, we define the Kleisli composition. First, for all objects $A$, $B$, and $C$, consider the following morphism:
\begin{align*}
&w = \Bigl(
((B \Rightarrow MC) \otimes (A \Rightarrow MB)) \otimes A
\xrightarrow{\cong}
(B \Rightarrow MC) \otimes ((A \Rightarrow MB) \otimes A)
\xrightarrow{\id \otimes \app}
(B \Rightarrow MC) \otimes MB
\\ & \hspace{9cm}
\xrightarrow{\tau'}
M((B \Rightarrow MC) \otimes B)
\xrightarrow{M\app}
MMC
\xrightarrow{\mu}
MC
\Bigr)
\end{align*}
We define the composition of Kleisli morphisms as $\kcomp = \mladj{w} : (B \Rightarrow MC)\otimes (A \Rightarrow MB) \to A \Rightarrow MC$, and the identity as $\mladj{I \otimes A \xrightarrow{\cong} A \xrightarrow{\eta} MA} : I \to A \Rightarrow MA$. The triple $\tuple{A\Rightarrow MA,\, \kcomp ,\, \kident}$ is a monoid.

\subsection{Symmetric Eilenberg--Moore monoids}

We now describe how the monad \verb+m [a]+ arises as a composition of adjoint functors. This is not a new construction, so we skip the proofs.
Consider the category $\catname{Mon}$ of monoids in a symmetric monoidal category $\catvar C$. Assume that the obvious forgetful functor $U_{\catname{Mon}} : \catname{Mon} \to \catvar C$ has a left adjoint $F_{\catname{Mon}} : \catvar C \to \catname{Mon}$. The induced monad $U_{\catname{Mon}} F_{\catname{Mon}}$ is the list monad.
Now, given a commutative monad $M$ on $\catvar C$, we define a monad $\overline M$ on $\catname{Mon}$ (in fact, a lifting in the sense of Beck~\cite{distr}). It is given as follows:
\begin{align*}
&\overline M \tuple{A, m, u} = \tuple{MA,\ \ MA \otimes MA \xrightarrow{\tau} M(A \otimes MA) \xrightarrow{M\tau'} MM(A \otimes A) \xrightarrow{\mu m} MA,\ \ I \xrightarrow{u} A \xrightarrow{\eta} MA}
\\
&
\overline M f = M f
\\
&
\mu^{\overline M} = \mu \qquad \eta^{\overline M} = \eta
\end{align*}
Let $F_{\,\overline M} \dashv U_{\,\overline M}$ be the Eilenberg--Moore adjunction of $\overline M$. Since adjoint functors compose, we obtain an adjunction $F_{\,\overline M} F_{\catname{Mon}} \dashv U_{\catname{Mon}} U_{\,\overline M}$. The induced monad $U_{\catname{Mon}} U_{\,\overline M} F_{\,\overline M} F_{\catname{Mon}}$ corresponds to the Haskell monad \verb+m [a]+. Here, we call this monad $\mathsf{CLT}_M$ (`commutative list transformer').

Now, we take a closer look at the Eilenberg--Moore category of  $\overline M$. Each object consists of a pair $\tuple{\tuple{A,m,u}, a}$, where $\tuple{A,m,u}$ is a monoid, and $a$ has the Eilenberg--Moore property, that is, $a \cdot Ma = a \cdot \mu$ and $a \cdot \eta = \id$ (since the monadic structure of $\overline M$ is identical to the monadic structure of $M$). Note that $a$ is a morphism in $\catname{Mon}$, so it preserves the monoid structure. That is, the following two diagrams commute:
\begin{equation}
\label{eq:diagsdalgd1}
\begin{tikzpicture}
\node(a0){$MA \otimes MA$};
\node(a1)[below of=a0]{$M(A \otimes MA)$};
\node(a2)[right of=a1, xshift=4em]{$MM(A \otimes A)$};
\node(b2)[right of=a2, xshift=4em]{$MA$};
\node(c0)[right of=a0, xshift=18em]{$A \otimes A$};
\node(c2)[right of=b2, xshift=2em]{$A$};
\path[arrows={-latex}, font=\scriptsize]
(a0) edge node [auto] {$\tau$} (a1)
(a1) edge node [auto] {$M\tau'$} (a2)
(a2) edge node [auto] {$\mu m$} (b2)
(b2) edge node [auto] {$a$} (c2)
(c0) edge node [auto] {$m$} (c2)
(a0) edge node [auto] {$a \otimes a$} (c0)
;
\end{tikzpicture}
\end{equation}
\begin{equation}
\label{eq:diagsdalgd2}
\begin{tikzpicture}
\node(a0){$I$};
\node(a1)[right of=a0, xshift=2em]{$A$};
\node(a2)[right of=a1, xshift=2em]{$MA$};
\node(b2)[below of=a0, xshift=0em]{$A$};
\path[arrows={-latex}, font=\scriptsize]
(a0) edge node [auto] {$u$} (a1)
(a1) edge node [auto] {$\eta$} (a2)
(a2) edge node [auto] {$a$} (b2)
(a0) edge node [auto] {$u$} (b2)
;
\end{tikzpicture}
\end{equation}
Note that the diagram~\eqref{eq:diagsdalgd2} commutes for all $a$ with the Eilenberg--Moore property. This, together with some rearranging of the elements of the tuples, leads us to the following equivalent definition of algebras for $\overline M$:
\begin{defn}
Let $M$ be a commutative monad on a symmetric monoidal category $\catvar C$. A \emph{symmetric Eilenberg--Moore $M$-monoid} is a tuple $\tuple{A,a,m,u}$, such that:
\begin{itemize}
\item $\tuple{A,a}$ is an Eilenberg--Moore algebra,
\item $\tuple{A,m,u}$ is a monoid,
\item coherence: the diagram~\eqref{eq:diagsdalgd1} commutes.
\end{itemize}
A morphism between two symmetric Eilenberg--Moore $M$-monoids is given by a $\catvar C$-morphism that is both a morphism between the Eilenberg--Moore algebra parts and the monoid parts. We call the category of symmetric Eilenberg--Moore $M$-monoids $\catname{SEMMon}_M$.
\end{defn}

The name is justified by the following theorem:

\begin{thm}
\label{thm:semmonisemmon}
Every symmetric Eilenberg--Moore $M$-monoid is an Eilenberg--Moore $M$-monoid.
\end{thm}

\begin{remark}
\label{rem:symem}
One could also imagine a `twisted' definition of Eilenberg--Moore monoids that uses $\tau'$ instead of $\tau$, and a coherence condition that equates the two corresponding morphisms $A \otimes MA \to A$. The proof of Theorem~\ref{thm:semmonisemmon} can be easily adapted to state that every symmetric Eilenberg--Moore monoid is a `twisted' Eilenberg--Moore monoid. Additionally, one can prove that a quadruple that is both an Eilenberg--Moore monoid and a `twisted' Eilenberg--Moore monoid for a commutative monad $M$ is necessarily a symmetric Eilenberg--Moore monoid.
\end{remark}

Since the definition of symmetric Eilenberg--Moore monoids is a simple rearrangement of the definition of algebras for $\mathsf{CLT}_M$, the adjunction $F_{\,\overline M} F_{\catname{Mon}} \dashv U_{\catname{Mon}} U_{\,\overline M}$ gives us that the obvious forgetful functor $U_{\catname{SEMMon}} : \catname{SEMMon}_M \to \catvar C$ has a left adjoint $F_{\catname{SEMMon}}$, and that the induced monad is equal to $\mathsf{CLT}_M$. Although a composition of two monadic adjunctions is not always monadic, it is so in this case (it follows form a general theorem of Beck about algebras for composite monads~\cite[Proposition 2]{distr}):

\begin{thm}
\label{thm:revisitingstrictlymonadic}
The adjunction $F_{\catname{SEMMon}} \dashv U_{\catname{SEMMon}}$ is strictly monadic. This entails that $\catname{SEMMon}_M$ is isomorphic to the category of algebras for $\mathsf{CLT}_M$.
\end{thm}

\subsection{Endomorphism representation and a continuation-based implementation}

Now, assume that $\catvar C$ is a closed symmetric monoidal category. We define a symmetric Eilenberg--Moore monoid of Kleisli endomorphisms, that is, objects $A \Rightarrow MA$:

\begin{thm}
\label{thm:revisitingdefofendo}
Let $\catvar C$ be a symmetric closed monoidal category, and $M$ be a commutative monad on $\catvar C$. Then, the tuple
\begin{equation*}
\tuple{
A \Rightarrow MA
,\ 
\mladj{q}
,\ 
\kcomp
,\ 
\kident
},
\end{equation*}
where
\begin{equation*}
q = \Bigl(
M(A \Rightarrow MA) \otimes A
\xrightarrow{\tau}
M((A \Rightarrow MA) \otimes A)
\xrightarrow{M\app}
MMA
\xrightarrow{\mu}
MA
\Bigr),
\end{equation*}
is a symmetric Eilenberg--Moore monoid.
\begin{equation*}
\end{equation*}
\end{thm}

It is left to prove that the monoid $A \Rightarrow MA$ is universal for the sufficient subcategory of symmetric Eilenberg--Moore monoids:

\begin{thm}
\label{thm:revisitingrepresentation}
For each symmetric Eilenberg--Moore monoid of the shape $\tuple{MA, \mu, m, u}$, the morphism $\mladj{s} : MA \to (A \Rightarrow MA)$, where
\begin{equation*}
s = \Bigl(
MA \otimes A \xrightarrow {\id \otimes \eta} MA \otimes MA \xrightarrow{m} MA
\Bigr),
\end{equation*}
has the following properties:
\begin{itemize}
\item it is a morphism $\tuple{MA, \mu, m, u} \to \tuple{A\Rightarrow MA,\,\mladj{q},\,\kcomp,\,\kident}$ between symmetric Eilenberg--Moore monoids,
\item it is a split monomorphism in $\catvar C$, that is, there exists a morphism $r : (A \Rightarrow MA) \to MA$ in $\catvar C$ such that $r \cdot \mladj{s} = \id$.
\end{itemize}
\end{thm}

The codensity monad that uses the representation above can be encoded in Haskell as follows:
\begin{Verbatim}
     data CLT m a = CLT (forall x. (a -> x -> m x) -> x -> m x)
\end{Verbatim}
Intuitively, this type represents folds over a list-like structure. Folding a single element can produce some effects in the monad \verb+m+, but it does not depend on the effects produced by previous elements (it can depend on the values though). The nil of the list (that is, the failure continuation) does not produce monadic effects on its own.

\section{Discussion}

\label{sec:coind}

Equations similar to the conditions in the definition of Eilenberg--Moore monoids were previously discussed by Hinze~\cite{DBLP:conf/icfp/Hinze00}, although in a different setting, that is, as equations between Haskell expressions. Jaskelioff and Moggi~\cite{Jaskelioff20104441} suggest that an equational theory like the one discussed in Section~\ref{sec:emmonoids} induces the list monad transformer, but they leave this without a proof.
Wand and Vaillancourt~\cite{Wand} use logical relations to compare two metalanguages with backtracking: one based on streams, and the other on a two-continuation monad. 
Eilenberg--Moore algebras of the resumption monad are also known as $F$-and-$M$-algebras. They were used by Filinski and St{\o}vring~\cite{Filinski:2007:IRE:1291220.1291168} (and later by Atkey \textit{et al.}~\cite{induction-with-effects,interleaving}) to model data structures that interleave pure data and effects.

In eager languages, the bare list monad is rarely used as a basis of backtracking computations, since the entire list structure is always computed upfront. Thus, in ML-like languages, one uses the type of lazy lists, which produce elements on demand. It can be implemented using the Haskell syntax as follows:
\begin{Verbatim}
     data LazyList a = LazyList (() -> Maybe (a, LazyList a))
\end{Verbatim}
Given a value of this type, one can force the next step of the computation by supplying the unit value~\verb+()+, and only then the structure is evaluated. This is nothing else than the list monad transformer (`done right') applied to the reader monad \verb+() -> a+.

Some languages provide separate primitives for inductive and coinductive data. From the point of view of semantics, it means that the language supports types given by initial algebras and final coalgebras separately. It is an interesting challenge for future work to describe the `coinductive' list monad transformer, given by carriers of final coalgebras. In such a case, the free monad becomes the \emph{free completely iterative monad} introduced by Aczel \textit{et al.}~\cite{DBLP:journals/tcs/AczelAMV03}, and the resumption monad becomes the \emph{coinductive resumption monad} described by Pir\'og and Gibbons~\cite{DBLP:journals/entcs/PirogG14}. The universal properties of both constructions are similar to those of their inductive counterparts, but considerably more complicated (see Ad\'amek \textit{et al.}~\cite{DBLP:journals/lmcs/AdamekMV06} for the case of the free completely iterative monad).

Another challenge for future work is to extend the current development with control operators, such as Prolog's \emph{cut} or fair disjunction. These features can be found, for example, in Kiselyov \textit{et al.}'s implementation~\cite{DBLP:conf/icfp/KiselyovSFS05}. We hope that such control structures can be obtained using the methods described in this paper.

\section*{Acknowledgements}

I would like to thank Tom Schrijvers for his remarks on an early draft of this paper, and the anonymous reviewers for their detailed comments and helpful suggestions.

\bibliographystyle{eptcs}

\appendix
\section{Proofs}

\subsection{Lemma~\ref{thm:emfreemonoidfromfmonoid}}

The conditions (1) and (2) from Definition~\ref{defn:emmonoids} are trivial. For (3), we need to show that the following diagram commutes:
\begin{equation*}
\begin{tikzpicture}
\node(a) {$G^*A \otimes A$};
\node(b)[below of=a] {$G^*(A \otimes A)$};
\node(c)[right of=a, xshift=10em] {$A \otimes A$};
\node(d)[right of=b, xshift=4em] {$G^*A$};
\node(e)[right of=d, xshift=2em] {$A$};
\path[arrows={-latex}, font=\scriptsize]
(a) edge node [auto] {$\foldfree g \otimes \id$} (c)
(a) edge node [auto] {$\widetilde\tau$} (b)
(b) edge node [auto] {$G^* m$} (d)
(d) edge node [auto] {$\foldfree g$} (e)
(c) edge node [auto] {$m$} (e)
;
\end{tikzpicture}
\end{equation*}
We show that both paths satisfy the universal property of strongly generated free monads, so they are equal. The top-right path:
\begin{equation*}
\begin{tikzpicture}
\node(a) {$G^*A \otimes A$};
\node(b)[right of=a, xshift=4em] {$A \otimes A$};
\node(c)[right of=b, xshift=12em] {$A$};
\node(x1)[above of=a, yshift=4em] {$GG^*A \otimes A$};
\node(x2)[right of=x1, xshift=4em] {$G(G^*A \otimes A)$};
\node(x3)[right of=x2, xshift=5em] {$G(A\otimes A)$};
\node(x4)[right of=x3, xshift=3em] {$GA$};
\node(y)[below of=x2] {$GA\otimes A$};
\node(z)[below of=a] {$A\otimes A$};
\path[arrows={-latex}, font=\scriptsize]
(a) edge node [auto] {$\foldfree g \otimes \id$} (b)
(b) edge node [auto] {$m$} (c)
(x1) edge node [auto] {$\cons \otimes \id$} (a)
(x1) edge node [auto] {$\tau$} (x2)
(x2) edge node [auto] {$G(\foldfree g \otimes \id)$} (x3)
(x3) edge node [auto] {$Gm$} (x4)
(x4) edge node [auto] {$g$} (c)
(y) edge node [auto] {$g \otimes \id$} (b)
(x1) edge node [auto,yshift=-0.4em] {$F\foldfree g \otimes \id$} (y)
(y) edge node [auto] {$\tau$} (x3)
(z) edge node [auto] {$\eta \otimes \id$} (a)
(z) edge node [auto, yshift=-0.2em] {$\id$} (b)
;
\node[below of=x1, xshift=4em, yshift=-1em] {\ding{192}};
\node[below of=x1, xshift=9em, yshift=2em] {\ding{193}};
\node[below of=x1, xshift=18em, yshift=0em] {\ding{194}};
\node[below of=x1, xshift=1.5em, yshift=-5.5em] {\ding{195}};
\end{tikzpicture}
\end{equation*}
\ding{192} $\foldfree g$ is an algebra morphism,
\ding{193} naturality of $\tau$,
\ding{194} $\tuple{A,g,m,u}$ is a $G$-monoid,
\ding{195} universal property of~$\foldfree\spholder$.

The left-bottom path:
\begin{equation*}
\begin{tikzpicture}
\node(a) {$G^*A \otimes A$};
\node(b)[right of=a, xshift=12em] {$G^*(A \otimes A)$};
\node(c)[right of=b, xshift=4em] {$G^*A$};
\node(d)[right of=c, xshift=4em] {$A$};
\node(x1)[above of=a, yshift=0em] {$GG^*A \otimes A$};
\node(x2)[right of=x1, xshift=4em] {$G(G^*A \otimes A)$};
\node(x3)[right of=x2, xshift=4em] {$GG^*(A\otimes A)$};
\node(x4)[right of=x3, xshift=4em] {$GG^*A$};
\node(x5)[right of=x4, xshift=4em] {$GA$};
\node(z)[below of=a] {$A\otimes A$};
\node(z1)[right of=z,xshift=20em, yshift=0em] {$A$};
\path[arrows={-latex}, font=\scriptsize]
(a) edge node [auto] {$\widetilde \tau$} (b)
(b) edge node [auto] {$G^*m$} (c)
(c) edge node [auto] {$\foldfree g$} (d)
(x1) edge node [auto] {$\cons$} (a)
(x1) edge node [auto] {$\tau$} (x2)
(x2) edge node [auto] {$G\widetilde\tau$} (x3)
(x3) edge node [auto] {$GG^*m$} (x4)
(x4) edge node [auto] {$G \foldfree g$} (x5)
(x5) edge node [auto] {$g$} (d)
(x4) edge node [auto] {$\cons$} (c)
(x3) edge node [auto] {$\cons$} (b)
(z) edge node [auto] {$\eta \otimes \id$} (a)
(z) edge node [auto, yshift=-0.2em] {$\eta$} (b)
(z) edge node [auto] {$m$} (z1)
(z1) edge node [auto] {$\eta$} (c)
(z1) edge node [auto, yshift=-0.2em] {$\id$} (d)
;
\node[below of=x1, xshift=8em, yshift=2em] {\ding{192}};
\node[below of=x1, xshift=20em, yshift=2em] {\ding{193}};
\node[below of=x1, xshift=28em, yshift=2em] {\ding{194}};
\node[below of=x1, xshift=2em, yshift=-2em] {\ding{195}};
\node[below of=x1, xshift=18em, yshift=-2em] {\ding{196}};
\node[below of=x1, xshift=25em, yshift=-2em] {\ding{197}};
\end{tikzpicture}
\end{equation*}
\ding{192} definition of strongly generated free monad,
\ding{193} naturality of $\cons$,
\ding{194} $\foldfree\spholder$ is a morphism of algebras,
\ding{195} properties of strength,
\ding{196} naturality of $\eta$,
\ding{197} universal property of $\foldfree\spholder$.

\subsection{Theorem~\ref{thm:adj}}

We split the proof into a number of lemmata. We need to show that $F$ is a functor (Lemma~\ref{lem:p:a}), that $\ladj\spholder$ is natural (Lemma~\ref{lem:p:b}), that $\radj \spholder$ produces morphisms between Eilenberg--Moore monoids (Lemma~\ref{lem:p:c}), and that $\ladj\spholder$ is an inverse of $\radj\spholder$ (Lemma~\ref{lem:p:d}).

\begin{lemma}
\label{lem:p:a}
The assignment $F$ is a functor.
\end{lemma}
\begin{proof}
First, we check that $FA$ is an Eilenberg--Moore monoid, that is, that the three conditions from Definition~\ref{defn:emmonoids} hold. The first two are obvious. For the third one, consider the following diagram, which is the desired coherence diagram with the definitions of $\mathfrak m$ and $\mathfrak u$ unfolded:
\begin{equation*}
\begin{tikzpicture}
\node(a) {$MM(\withm A)^* \MID \otimes M(\withm A)^* \MID$};
\node(b)[right of=a, xshift=16em] {$M(\withm A)^* \MID \otimes M(\withm A)^* \MID$};
\node(c)[below of=a]{$M(M(\withm A)^* \MID \otimes M(\withm A)^* \MID)$};
\node(d)[below of=c]{$MM((\withm A)^* \MID \otimes M(\withm A)^* \MID)$};
\node(e)[below of=d]{$MM(\withm A)^* (\MID \otimes M(\withm A)^* \MID)$};
\node(f)[below of=e]{$MM(\withm A)^* M(\withm A)^* \MID$};
\node(g)[below of=b, yshift=-4em]{$M((\withm A)^* \MID \otimes M(\withm A)^* \MID)$};
\node(h)[below of=g]{$M(\withm A)^* (\MID \otimes M(\withm A)^* \MID)$};
\node(i)[below of=h]{$M(\withm A)^* M(\withm A)^* \MID$};
\node(j)[below of=i]{$M(\withm A)^* \MID$};
\node(k)[below of=f]{$MM(\withm A)^* \MID$};
\path[arrows={-latex}, font=\scriptsize]
(a) edge node [auto] {$\mu^M \otimes \id$} (b)
(a) edge node [auto] {$\tau^M$} (c)
(c) edge node [auto] {$M\tau^M$} (d)
(d) edge node [auto] {$MM\spanst{\tau^M}$} (e)
(e) edge node [auto] {$\cong$} (f)
(b) edge node [auto] {$\tau^M$} (g)
(g) edge node [auto] {$M\spanst{\tau^M}$} (h)
(h) edge node [auto] {$\cong$} (i)
(i) edge node [auto] {$\mu^\res$} (j)
(f) edge node [auto] {$M\mu^\res$} (k)
(k) edge node [auto] {$\mu^M$} (j)
(d) edge node [auto] {$\mu^M$} (g)
(f) edge node [auto] {$\mu^M$} (i)
;
\node[below of=a, xshift=10em, yshift=0em] {\ding{192}};
\node[below of=a, xshift=10em, yshift=-8em] {\ding{193}};
\node[below of=a, xshift=10em, yshift=-13.5em] {\ding{194}};
\end{tikzpicture}
\end{equation*}
\ding{192}
properties of strength,
\ding{193}
naturality of $\mu^M$,
\ding{194}
$\mu^\res$ is defined via a distributive law.

To verify the morphism part, let $f : A \to B$ be a morphism in $\catvar C$. It is trivial that $Ff$ is a morphism between Eilenberg--Moore algebra parts of $FA$ and $FB$, as it amounts to the naturality of $\mu^M$. As for the monoid parts, the preservation of the unit is simply the fact that $M(\withm f)^*I$ is a monad morphism. For the preservation of the multiplication, consider the following diagram:
\begin{equation*}
\begin{tikzpicture}
\node(a) {$M(\withm A)^*I \otimes M(\withm A)^*I$};
\node(b)[below of=a] {$M((\withm A)^*I \otimes M(\withm A)^*I)$};
\node(c)[below of=b] {$M(\withm A)^*(I \otimes M(\withm A)^*I)$};
\node(d)[below of=c] {$M(\withm A)^*M(\withm A)^*I$};
\node(e)[below of=d] {$M(\withm A)^*I$};
\node(ar)[right of=a, xshift=20em] {$M(\withm A)^*I \otimes M(\withm A)^*I$};
\node(br)[below of=ar] {$M((\withm B)^*I \otimes M(\withm B)^*I)$};
\node(cr)[below of=br] {$M(\withm B)^*(I \otimes M(\withm B)^*I)$};
\node(dr)[below of=cr] {$M(\withm B)^*M(\withm B)^*I$};
\node(er)[below of=dr] {$M(\withm B)^*I$};
\path[arrows={-latex}, font=\scriptsize]
(a) edge node [auto] {$\tau^M$} (b)
(b) edge node [auto] {$M\widetilde{\tau^M}$} (c)
(c) edge node [auto] {$\cong$} (d)
(d) edge node [auto] {$\mu^{\mathcal R}$} (e)
(ar) edge node [auto] {$\tau^M$} (br)
(br) edge node [auto] {$M\widetilde{\tau^M}$} (cr)
(cr) edge node [auto] {$\cong$} (dr)
(dr) edge node [auto] {$\mu^{\mathcal R}$} (er)
(a) edge node [auto] {$M(\withm f)^*I \otimes M(\withm f)^*I$} (ar)
(b) edge node [auto] {$M((\withm f)^*I \otimes M(\withm f)^*I)$} (br)
(c) edge node [auto] {$M(\withm f)^* (\id \otimes M(\withm f)^*)$} (cr)
(d) edge node [auto] {$M(\withm f)^*M(\withm f)^*$} (dr)
(e) edge node [auto] {$M(\withm f)^*$} (er)
;
\node[below of=a, xshift=12em, yshift=2.5em] {\ding{192}};
\node[below of=b, xshift=12em, yshift=2.5em] {\ding{193}};
\node[below of=c, xshift=12em, yshift=2.5em] {\ding{194}};
\node[below of=d, xshift=12em, yshift=2.5em] {\ding{195}};
\end{tikzpicture}
\end{equation*}
\ding{192} naturality of $\tau$,
\ding{193} the fact that $(\withm f)^*$ preserves strength, and naturality o $\widetilde\tau$,
\ding{194} naturality of $\cong$,
\ding{195} monad morphism.
\end{proof}

\begin{lemma}
\label{lem:p:b}
The assignment $\ladj\spholder$ is a natural transformation.
\end{lemma}
\begin{proof}
Let $f : FA \to \tuple{B,b,m^B,u^B}$ and $l : \tuple{B,b,m^B,u^B} \to \tuple{Y,y,m^Y,u^Y}$ be morphisms in $\catname{EMMon}_M$, and $r : X \to A$ be a morphism in $\catvar C$. The following diagram commutes, where the top-most path is equal to $\ladj{l \cdot f \cdot Fr}$, while the bottom-most path is equal to $Ul \cdot \ladj f \cdot r$:
\begin{equation*}
\begin{tikzpicture}
\node(a){$X$};
\node(b1)[right of=a, yshift=-2em, xshift=4em]{$A$};
\node(b2)[right of=a, yshift=2em, xshift=4em]{$X \otimes I$};
\node(c1)[right of=b1, yshift=0em, xshift=4em]{$A \otimes I$};
\node(c2)[right of=b2, yshift=0em, xshift=4em]{$M(\withm X)^*I$};
\node(d)[right of=c1, yshift=2em, xshift=4em]{$M(\withm A)^*I$};
\node(e)[right of=d, yshift=0em, xshift=2em]{$B$};
\node(f)[right of=e, yshift=0em, xshift=0em]{$Y$};
\path[arrows={-latex}, font=\scriptsize]
(a) edge node [below, xshift=-0.3em] {$r$} (b1)
(a) edge node [above] {$\cong$} (b2)
(b1) edge node [below] {$\cong$} (c1)
(b2) edge node [above] {$\INR$} (c2)
(c1) edge node [below, xshift=0.5em] {$\INR$} (d)
(c2) edge node [above, xshift=1.5em] {$M(\withm r)^*I$} (d)
(d) edge node [auto] {$f$} (e)
(e) edge node [auto] {$l$} (f)
(b2) edge node [above, xshift=0.7em] {$r \otimes \id$} (c1)
;  
\node[right of=a, xshift=4em] {\ding{192}};
\node[right of=a, xshift=12em] {\ding{193}};
\end{tikzpicture}
\end{equation*}
\ding{192} naturality of $\cong$, \ding{193} definition of $\INR$.
\end{proof}

\begin{lemma}
\label{thm:fromemmonoidtomonoidaotimesd882jf}
Given an Eilenberg--Moore $M$-monoid $\tuple{B,b,m^B,u^B}$ and a $\catvar C$-morphism $g : A \to B$, the tuple $\tuple{B ,\ A \otimes MB \xrightarrow{g \otimes b} B \otimes B \xrightarrow{m^B} B ,\ m^B ,\ u^B}$ is an $(A \otimes M\pholder)$-monoid for $A \otimes M\pholder$ understood as a functor with strength given as $(A \otimes MX) \otimes Y \xrightarrow{\cong} A \otimes (MX \otimes Y) \xrightarrow{\id \otimes \tau} A \otimes M(X \otimes Y)$.
\end{lemma}
\begin{proof}
Since $\tuple{B,m^B,u^B}$ is a monoid by definition, it is left to check the coherence diagram:
\begin{equation*}
\begin{tikzpicture}
\node(a1){$(A\otimes MB)\otimes B$};
\node(a2)[right of=a1, xshift=7em]{$(B\otimes MB)\otimes B$};
\node(a3)[right of=a2, xshift=7em]{$(B\otimes B)\otimes B$};
\node(a4)[right of=a3, xshift=7em]{$B\otimes B$};
\node(b1)[below of=a1]{$A\otimes (MB\otimes B)$};
\node(b2)[right of=b1, xshift=7em]{$B\otimes (MB\otimes B)$};
\node(b3)[right of=b2, xshift=7em]{$B\otimes (B\otimes B)$};
\node(c1)[below of=b1]{$A\otimes M(B\otimes B)$};
\node(c2)[right of=c1, xshift=7em]{$B\otimes M(B\otimes B)$};
\node(d1)[below of=c1]{$A \otimes MB$};
\node(d2)[right of=d1, xshift=7em]{$B \otimes MB$};
\node(d3)[right of=d2, xshift=7em]{$B \otimes B$};
\node(d4)[right of=d3, xshift=7em]{$B$};
\path[arrows={-latex}, font=\scriptsize]
(a1) edge node [auto] {$(g \otimes \id)\otimes \id$} (a2)
(a2) edge node [auto] {$(\id \otimes b)\otimes \id$} (a3)
(a3) edge node [auto] {$m^B \otimes \id$} (a4)
(b1) edge node [auto] {$g \otimes (\id\otimes \id)$} (b2)
(b2) edge node [auto] {$\id \otimes (b\otimes \id)$} (b3)
(a1) edge node [auto] {$\cong$} (b1)
(a2) edge node [auto] {$\cong$} (b2)
(a3) edge node [auto] {$\cong$} (b3)
(b1) edge node [auto] {$\id \otimes \tau$} (c1)
(b2) edge node [auto] {$\id \otimes \tau$} (c2)
(c1) edge node [auto] {$\id \otimes Mm^B$} (d1)
(c2) edge node [auto] {$\id \otimes Mm^B$} (d2)
(b3) edge node [auto] {$\id \otimes m^B$} (d3)
(a4) edge node [auto] {$m^B$} (d4)
(d1) edge node [auto] {$g \otimes \id$} (d2)
(d2) edge node [auto] {$\id \otimes b$} (d3)
(d3) edge node [auto] {$m^B$} (d4)
;
\node[right of=a1, yshift=-1.5em, xshift=1.5em] {\ding{192}};
\node[right of=a1, yshift=-1.5em, xshift=12.5em] {\ding{193}};
\node[right of=a1, yshift=-4.5em, xshift=24em] {\ding{194}};
\node[right of=a1, yshift=-8em, xshift=1.5em] {\ding{195}};
\node[right of=a1, yshift=-8em, xshift=12.5em] {\ding{196}};
\end{tikzpicture}
\end{equation*}
\ding{192} and
\ding{193} naturality of $\cong$,
\ding{194} associativity of $m^B$,
\ding{195} $\otimes$ is a bifunctor,
\ding{196} coherence of $\tuple{B,b,m^B,u^B}$.
\end{proof}

\begin{lemma}
\label{lem:p:c}
For a $\catvar C$-morphism $g : A \to U\tuple{B,b,m^B,u^B}$, the morphism $\radj g$ is an $\catname{EMMon}_M$-morphism of the type $FA \to \tuple{B,b,m^B,u^B}$ .
\end{lemma}
\begin{proof}
The fact that $\radj g$ is a morphism between the Eilenberg--Moore algebra parts follows from the following diagram, where the right-most edge is equal to $\radj g$ unfolded as in Lemma~\ref{thm:foldreseqprops}(1):

\begin{equation*}
\begin{tikzpicture}
\node(a1){$MM(\withm A)^*I$};
\node(a2)[right of=a1, xshift=12em]{$M(\withm A)^*I$};
\node(b1)[below of=a1]{$MM(\withm A)^*B$};
\node(b2)[below of=a2]{$M(\withm A)^*B$};
\node(c1)[below of=b1,xshift=4em]{$MMB$};
\node(c2)[below of=b2,xshift=-4em]{$MB$};
\node(d1)[below of=c1,xshift=-4em]{$MB$};
\node(d2)[below of=c2,xshift=4em]{$B$};
\path[arrows={-latex}, font=\scriptsize]
(a1) edge node [auto] {$\mu^M$} (a2)
(b1) edge node [auto] {$\mu^M$} (b2)
(a1) edge node [left] {$MM(\withm A)^*u^B$} (b1)
(a2) edge node [auto] {$M(\withm A)^*u^B$} (b2)
(c1) edge node [auto] {$\mu^M$} (c2)
(b1) edge node [left] {$M\foldres{b,\ m^B \cdot(g \otimes \id),\ \id}$} (d1)
(b2) edge node [auto] {$\foldres{b,\ m^B \cdot(g \otimes \id),\ \id}$} (d2)
(d1) edge node [auto] {$b$} (d2)
(c2) edge node [auto] {$b$} (d2)
(c1) edge node [auto] {$Mb$} (d1)
(b1) edge node [right] {$\ MM\foldfree{m^B \cdot (g \otimes b)}$} (c1)
(b2) edge node [left] {$M\foldfree{m^B \cdot (g \otimes b)}\ $} (c2)
;
\node[right of=a1, yshift=-1.5em, xshift=4em] {\ding{192}};
\node[right of=a1, yshift=-4.8em, xshift=4em] {\ding{193}};
\node[right of=a1, yshift=-8em, xshift=-3em] {\ding{194}};
\node[right of=a1, yshift=-8em, xshift=11em] {\ding{195}};
\node[right of=a1, yshift=-9.5em, xshift=4em] {\ding{196}};
\end{tikzpicture}
\end{equation*}
\ding{192} and
\ding{193} naturality of $\mu^M$,
\ding{194} and
\ding{195} Lemma~\ref{thm:foldreseqprops}(2),
\ding{196} $b$ has the Eilenberg--Moore property.

The fact that $\radj g$ commutes with monoid multiplication follows from the following diagram, in which the left-most edge is the definition of $\mathfrak m$ (that is, the multiplication of $FA$):

\begin{equation}
\label{eq:bigdiagradjg}
\begin{tikzpicture}
\node(a1){$M(\withm A)^*I \otimes M(\withm A)^*I$};
\node(a2)[xshift=12em, right of=a1]{$M(\withm A)^*I \otimes B$};
\node(a3)[xshift=8em, right of=a2]{$B \otimes B$};
\node(b1)[below of=a1]{$M((\withm A)^*I \otimes M(\withm A)^*I)$};
\node(b2)[xshift=12em, right of=b1]{$M((\withm A)^*I \otimes B)$};
\node(c1)[below of=b1]{$M(\withm A)^*(I \otimes M(\withm A)^*I)$};
\node(c2)[xshift=12em, right of=c1]{$M(\withm A)^*(I \otimes B)$};
\node(d1)[below of=c1]{$M(\withm A)^*M(\withm A)^*I$};
\node(d2)[xshift=12em, right of=d1]{$M(\withm A)^*B$};
\node(e1)[below of=d1]{$M(\withm A)^*I$};
\node(e3)[right of=e1, xshift=24em]{$B$};
\path[arrows={-latex}, font=\scriptsize]
(a1) edge node [auto] {$\id \otimes \radj g$} (a2)
(a2) edge node [auto] {$\radj g \otimes \id$} (a3)
(b1) edge node [auto] {$M(\id \otimes \radj g)$} (b2)
(c1) edge node [auto] {$M(\withm A)^*(\id \otimes \radj g)$} (c2)
(a1) edge node [auto] {$\tau$} (b1)
(a2) edge node [auto] {$\tau$} (b2)
(b1) edge node [auto] {$M\widetilde\tau$} (c1)
(b2) edge node [auto] {$M\widetilde\tau$} (c2)
(c1) edge node [auto] {$\cong$} (d1)
(e1) edge node [auto] {$\radj g$} (e3)
(d1) edge node [auto] {$\mu^\res$} (e1)
(d1) edge node [auto] {$M(\withm A)^*\radj g$} (d2)
(c2) edge node [auto] {$\cong$} (d2)
(a3) edge node [auto] {$m^B$} (e3)
(d2) edge node [auto,xshift=-2.5em,yshift=0.5em] {$\foldres{b,\ m^B \cdot (g\otimes\id) ,\ \id}$} (e3)
;
\node[right of=a1, yshift=-1.5em, xshift=5em] {\ding{192}};
\node[right of=a1, yshift=-5.5em, xshift=5em] {\ding{193}};
\node[right of=a1, yshift=-9.5em, xshift=5em] {\ding{194}};
\node[right of=a1, yshift=-13.5em, xshift=5em] {\ding{195}};
\node[right of=a1, yshift=-5.5em, xshift=19em] {\ding{196}};
\end{tikzpicture}
\end{equation}
\ding{192} naturality of $\tau$,
\ding{193} naturality of $\widetilde\tau$,
\ding{194} naturality of $\cong$,
\ding{195} see below,
\ding{196} see below.

To verify that the square \ding{195} above commutes, we unfold $\radj g$ as in Lemma~\ref{thm:foldreseqprops}(1). The desired diagram is then as follows:
\begin{equation*}
\begin{tikzpicture}
\node(a1){$M(\withm A)^*M(\withm A)^*I$};
\node(a2)[right of=a1,xshift=13em]{$M(\withm A)^*M(\withm A)^*B$};
\node(a3)[xshift=13em, right of=a2]{$M(\withm A)^*B$};
\node(b1)[below of=a1]{$M(\withm A)^*I$};
\node(b2)[below of=a2]{$M(\withm A)^*B$};
\node(b3)[below of=a3]{$B$};
\path[arrows={-latex}, font=\scriptsize]
(a1) edge node [auto] {$M(\withm A)^*M(\withm A)^* u^B$} (a2)
(a2) edge node [auto] {$M(\withm A)^*\foldres{b,\ m^B\cdot(g\otimes\id),\ \id}$} (a3)
(b1) edge node [auto] {$M(\withm A)^* u^B$} (b2)
(b2) edge node [auto] {$\foldres{b,\ m^B \cdot (g \otimes \id),\ \id}$} (b3)
(a1) edge node [auto] {$\mu^\res$} (b1)
(a2) edge node [auto] {$\mu^\res$} (b2)
(a3) edge node [left,yshift=0.5em] {$\foldres{b,\ m^B \cdot (g \otimes \id) ,\ \id}$} (b3)
;
\node[right of=a1, yshift=-1.5em, xshift=4em] {\ding{192}};
\node[right of=a1, yshift=-1.5em, xshift=18em] {\ding{193}};
\end{tikzpicture}
\end{equation*}
\ding{192} naturality of $\mu^\res$,
\ding{193} $\foldres{\spholder, \spholder, \id}$ has the Eilenberg--Moore property (Lemma~\ref{thm:foldreseqprops}(2)).

Below, we detail \ding{196} from the diagram~\eqref{eq:bigdiagradjg}:
\begin{equation*}
\begin{tikzpicture}
\node(a0){$M(\withm A)^*I \otimes B$};
\node(a1)[right of=a0, xshift=10em]{$M(\withm A)^*B \otimes B$};
\node(b0)[below of=a0, yshift=-4em]{$M((\withm A)^*I \otimes B)$};
\node(b1)[right of=b0, xshift=10em]{$M((\withm A)^*B \otimes B)$};
\node(c0)[below of=b0]{$M(\withm A)^*(I \otimes B)$};
\node(c1)[right of=c0, xshift=10em]{$M(\withm A)^*(B \otimes B)$};
\node(d0)[below of=c0]{$M(\withm A)^*B$};
\node(a3)[right of=a1, xshift=18em]{$B\otimes B$};
\node(aa2)[below of=a1, xshift=16em]{$MB\otimes B$};
\node(b2)[below of=aa2]{$M(B\otimes B)$};
\node(c2)[below of=b2]{$MB$};
\node(d3)[below of=a3, yshift=-12em]{$B$};
\path[arrows={-latex}, font=\scriptsize]
(a0) edge node [auto] {$M(\withm A)^* u^B \otimes \id$} (a1)
(b0) edge node [auto] {$M((\withm A)^* u^B \otimes \id)$} (b1)
(c0) edge node [auto] {$M(\withm A)^* (u^B \otimes \id)$} (c1)
(a0) edge node [auto] {$\tau$} (b0)
(a1) edge node [auto] {$\tau$} (b1)
(b0) edge node [auto] {$M\widetilde\tau$} (c0)
(b1) edge node [auto] {$M\widetilde\tau$} (c1)
(c0) edge node [auto] {$\cong$} (d0)
(c1) edge node [left, yshift=0.2em] {$M(\withm A)^*m^B$} (d0)
(a1) edge node [auto] {$\foldres{b ,\ m^B\cdot (g\otimes \id) ,\ \id} \otimes \id$} (a3)
(a1) edge node [auto, xshift=-1.9em, yshift=0.1em] {$M\foldfree{m^B \cdot (g \otimes b)} \otimes \id$} (aa2)
(aa2) edge node [right, xshift=-1em, yshift=-1em] {$b \otimes \id$} (a3)
(aa2) edge node [auto] {$\tau$} (b2)
(b2) edge node [auto] {$Mm^B$} (c2)
(c2) edge node [auto] {$b$} (d3)
(a3) edge node [auto] {$m^B$} (d3)
(d0) edge node [auto,xshift=6em] {$\foldres{b,\ m^B\cdot (g \otimes \id),\ \id}$} (d3)
(b1) edge node [auto] {$M(\foldfree{m^B \cdot (g \otimes b)} \otimes \id)$} (b2)
(d0) edge node [auto, yshift=1.2em, xshift=10em] {$M\foldfree{m^B \cdot (g \otimes b)}$} (c2)
;
\node[right of=a0, yshift=-3.5em, xshift=3em] {\ding{192}};
\node[right of=a0, yshift=-3.5em, xshift=15em] {\ding{193}};
\node[right of=a0, yshift=-1.5em, xshift=26em] {\ding{194}};
\node[right of=a0, yshift=-5.5em, xshift=30em] {\ding{195}};
\node[right of=a0, yshift=-9.5em, xshift=3em] {\ding{196}};
\node[right of=a0, yshift=-9.5em, xshift=19em] {\ding{197}};
\node[right of=a0, yshift=-12.75em, xshift=4em] {\ding{198}};
\node[right of=a0, yshift=-14em, xshift=26em] {\ding{199}};
\end{tikzpicture}
\end{equation*}
\ding{192} and
\ding{193} naturality of $\tau$,
\ding{194} Lemma~\ref{thm:foldreseqprops}(2),
\ding{195} coherence for $\tuple{B,b,m^B,u^B}$,
\ding{196} naturality of $\widetilde\tau$,
\ding{197} $M$-image of the coherence diagram for the Eilenberg--Moore $(\withm A)^*$-monoid generated as in Lemma~\ref{thm:emfreemonoidfromfmonoid} by the $(\withm A)$-monoid $\tuple{B ,\ m^B\cdot(g\otimes b) ,\ m^B,\ u^B}$ described in Lemma~\ref{thm:fromemmonoidtomonoidaotimesd882jf},
\ding{198} monoid,
\ding{199} Lemma~\ref{thm:foldreseqprops}(2).

It is left to show that $\radj g$ preserves units of the monoid parts of the respective Eilenberg--Moore monoids. The following diagram commutes, where the right-most edge is equal to $\radj g$ via Lemma~\ref{thm:foldreseqprops}(1):
\begin{equation*}
\begin{tikzpicture}
\node(a){$B$};
\node(b)[right of=a, xshift=4em]{$M(\withm A)^*B$};
\node(c)[below of=b]{$B$};
\node(x)[above of=a]{$I$};
\node(y)[above of=b]{$M(\withm A)^*I$};
\path[arrows={-latex}, font=\scriptsize]
(a) edge node [auto] {$\eta^\res$} (b)
(a) edge node [below] {$\id$} (c)
(b) edge node [auto] {$\foldres{b,\ m^B\cdot(g\otimes\id),\ \id}$} (c)
(x) edge node [auto] {$\mathfrak u\ =\ \eta^\res$} (y)
(x) edge node [auto] {$u^B$} (a)
(y) edge node [auto] {$M(\withm A)^*u^B$} (b)
;
\node[right of=x, yshift=-1.5em, xshift=0em] {\ding{192}};
\node[right of=x, yshift=-5.5em, xshift=2em] {\ding{193}};
\end{tikzpicture}
\end{equation*}
\ding{192} naturality of $\eta^\res$,
\ding{193} $\foldres{\spholder, \spholder, \id}$ has the Eilenberg--Moore property (Lemma~\ref{thm:foldreseqprops}(2)).
\end{proof}

\begin{lemma}
\label{lem:p:d}
The natural transformation $\ladj\spholder$ is a natural isomorphism with the inverse given by $\radj\spholder$.
\end{lemma}
\begin{proof}
In one direction, let $g : A \to U\tuple{B,b,m^B,u^B}$ be a $\catvar C$-morphism. The fact that $\ladj {\radj g} = g$ follows from the following diagram, where the longer path of the perimeter is obtained by unfolding the definitions of $\ladj \spholder$ and $\radj \spholder$:
\begin{equation*}
\begin{tikzpicture}
\node(a){$A$};
\node(b)[right of=a, xshift=20em]{$B$};
\node(x1)[below of=a, xshift=4em, yshift=1em]{$B$};
\node(x2)[right of=x1, xshift=4em]{$B \otimes I$};
\node(x3)[right of=x2, xshift=4em]{$B \otimes B$};
\node(y)[below of=x3, yshift=1em]{$A \otimes B$};
\node(c)[below of=a, yshift=-5em]{$A \otimes I$};
\node(d)[below of=b, yshift=-5em]{$M(\withm A)^*I$};
\path[arrows={-latex}, font=\scriptsize]
(a) edge node [auto] {$g$} (b)
(a) edge node [left] {$\cong$} (c)
(c) edge node [above] {$\INR$} (d)
(d) edge node [right] {$\foldres{b,\ m^B \cdot (g\otimes\id) ,\ u^B}$} (b)
(a) edge node [auto,yshift=-0.4em] {$g$} (x1)
(x1) edge node [auto] {$\cong$} (x2)
(x2) edge node [auto] {$\id \otimes u^B$} (x3)
(x3) edge node [auto,yshift=-0.4em] {$m^B$} (b)
(c) edge node [auto] {$g \otimes \id$} (x2)
(c) edge node [auto] {$\id \otimes u^B$} (y)
(y) edge node [auto,yshift=-0.2em] {$g \otimes \id$} (x3)
;
\node[below of=a, xshift=12em, yshift=3.0em] {\ding{192}};
\node[below of=a, xshift=2em, yshift=0.0em] {\ding{193}};
\node[below of=a, xshift=14em, yshift=-1.0em] {\ding{194}};
\node[below of=a, xshift=22em, yshift=-1.0em] {\ding{195}};
\end{tikzpicture}
\end{equation*}
\ding{192} right unit law for monoids,
\ding{193} naturality of $\cong$,
\ding{194} $\otimes$ is a bifunctor,
\ding{195} cancellation property of $\foldres{\spholder}$ (Lemma~\ref{thm:foldreseqprops}(3)).

In the other direction, we need to show that $\radj{\ladj f} = f$ for $f : FA \to \tuple{B,b,m^B,u^B}$. Unfolding the definition of $\radj \spholder$, we obtain:
\begin{equation*}
\radj{\ladj f} = \foldres{MB \xrightarrow{b} B, \ \ 
A \otimes B \xrightarrow{\ladj f \otimes \id} B \otimes B \xrightarrow{m^B} B, \ \  
I \xrightarrow{u^B} B}
\end{equation*}
To prove that it is equal to $f$, it is enough to show that the universal property of $\foldres{\spholder}$ (given in Lemma~\ref{thm:foldresuniprop}) holds for $f$.

The left-hand side of the first diagram in Lemma~\ref{thm:foldresuniprop} commutes, since morphisms in $\catname{EMMon}_M$ are necessarily morphisms between Eilenberg--Moore algebra parts. The right-hand side of the diagram (that is, the fact that $f$ is a morphisms between $(A \otimes \pholder)$-algebras) is given by the following diagram, in which we unfold the definition of $\ladj f$ in $m^B \cdot (\ladj f \otimes \id)$ (the bottom edge):
\begin{equation*}
\begin{tikzpicture}
\node(x){$A \otimes M(\withm A)^*I$};
\node(y)[right of=x, xshift=30em]{$M(\withm A)^*I$};
\node(xx)[below of=x, xshift=8em]{$(A \otimes I) \otimes M(\withm A)^*I$};
\node(yy)[right of=xx, xshift=15em]{$M(\withm A)^*I \otimes M(\withm A)^*I$};
\node(a)[below of=x, yshift=-4em]{$A \otimes B$};
\node(b)[right of=a, xshift=3em]{$(A \otimes I) \otimes B$};
\node(c)[right of=b, xshift=7em]{$M(\withm A)^*I \otimes B$};
\node(d)[right of=c, xshift=5em]{$B \otimes B$};
\node(e)[right of=d, xshift=3em]{$B$};
\path[arrows={-latex}, font=\scriptsize]
(a) edge node [auto] {${\cong} \otimes \id$} (b)
(b) edge node [auto] {$\INR \otimes \id$} (c)
(c) edge node [auto] {$f \otimes \id$} (d)
(d) edge node [auto] {$m^B$} (e)
(y) edge node [auto] {$f$} (e)
(x) edge node [auto] {$\id \otimes f$} (a)
(x) edge node [auto] {$\mu^\res \cdot \INR$} (y)
(x) edge node [auto, xshift=-0.4em] {${\cong} \otimes \id$} (xx)
(xx) edge node [auto] {$\INR \otimes \id$} (yy)
(yy) edge node [auto] {$f \otimes f$} (d)
(yy) edge node [auto] {$\mathfrak m$} (y)
;
\node[below of=x, xshift=17em, yshift=2.5em] {\ding{192}};
\node[below of=x, xshift=17em, yshift=-2em] {\ding{193}};
\node[below of=x, xshift=31em, yshift=-2em] {\ding{194}};
\end{tikzpicture}
\end{equation*}
\ding{192} see below, \ding{193} $\otimes$ is a bifunctor, \ding{194} $f$ is a monoid morphism (as a morphism in $\catname{EMMon}_M$).

Below, we detail \ding{192} from the diagram above. We unfold the definitions of $\INR$ and $\mathfrak m$.

\begin{equation*}
\begin{tikzpicture}
\node(a1) {$A \otimes M(\withm A)^*I$};
\node(a2)[right of=a1,xshift=6em] {$A \otimes MM(\withm A)^*I$};
\node(a3)[right of=a2,xshift=6em] {$(\withm A)^*M(\withm A)^*I$};
\node(a4)[right of=a3,xshift=7em] {$M(\withm A)^*M(\withm A)^*I$};
\node(aa4)[above of=a4] {$M(\withm A)^*I$};
\node(b1)[below of=a1] {$(A \otimes I) \otimes M(\withm A)^*I$};
\node(c1)[below of=b1] {$(A \otimes MI) \otimes M(\withm A)^*I$};
\node(d1)[below of=c1, yshift=-4em] {$(\withm A)^*I \otimes M(\withm A)^*I$};
\node(e1)[below of=d1] {$M(\withm A)^*I \otimes M(\withm A)^*I$};
\node(e4)[below of=a4, yshift=-16em] {$M((\withm A)^* I \otimes M(\withm A)^*I)$};
\node(d4)[below of=a4, yshift=-12em] {$M(\withm A)^* (I \otimes M(\withm A)^*I)$};
\node(x1)[below of=a2, xshift=2em]{$A \otimes (I \otimes M(\withm A)^*I)$};
\node(x2)[below of=x1, xshift=0em]{$A \otimes (MI \otimes M(\withm A)^*I)$};
\node(x3)[below of=x2, xshift=0em]{$A \otimes M(I \otimes M(\withm A)^*I)$};
\node(x4)[below of=x3, xshift=4em]{$(\withm A)^*(I \otimes M(\withm A)^*I)$};
\path[arrows={-latex}, font=\scriptsize]
(a4) edge node [auto] {$\mu^\res$} (aa4)
(a1) edge node [auto] {$\id \otimes \eta^M$} (a2)
(a2) edge node [auto] {$\emb$} (a3)
(a3) edge node [auto] {$\eta^M$} (a4)
(a1) edge node [auto] {${\cong} \otimes \id$} (b1)
(b1) edge node [auto] {$(\id \otimes \eta^M) \otimes \id$} (c1)
(c1) edge node [auto] {$\emb \otimes \id$} (d1)
(d1) edge node [auto] {$\eta^M \otimes \id$} (e1)
(d4) edge node [auto] {${\cong}$} (a4)
(e4) edge node [auto] {$M \widetilde \tau$} (d4)
(e1) edge node [auto] {$\tau$} (e4)
(d1) edge node [auto] {$\eta^M$} (e4)
(x1) edge node [auto] {$\id \otimes (\eta^M \otimes \id)$} (x2)
(x2) edge node [auto] {$\id \otimes \tau$} (x3)
(x3) edge node [auto] {$\emb$} (x4)
(a1) edge node [auto, yshift=-0.2em] {$\id \otimes {\cong}$} (x1)
(b1) edge node [auto] {$\cong$} (x1)
(c1) edge node [auto] {$\cong$} (x2)
(d1) edge node [auto] {$\widetilde \tau$} (x4)
(x4) edge node [auto] {$\eta^M$} (d4)
(x1.east) edge[bend left=60] node [auto] {$\id \otimes \eta^M$} (x3.east)
(x4) edge[bend right=60] node [auto] {$\cong$} (a3)
;
\node[below of=x, xshift=4em, yshift=1.5em] {\ding{192}};
\node[below of=x, xshift=8em, yshift=-2em] {\ding{193}};
\node[below of=x, xshift=6em, yshift=-6em] {\ding{194}};
\node[below of=x, xshift=18em, yshift=-4em] {\ding{195}};
\node[below of=x, xshift=21em, yshift=-1em] {\ding{196}};
\node[below of=x, xshift=26em, yshift=-2em] {\ding{197}};
\node[below of=x, xshift=6em, yshift=-14em] {\ding{198}};
\node[below of=x, xshift=26em, yshift=-14em] {\ding{199}};
\end{tikzpicture}
\end{equation*}
\ding{192} monoidal category,
\ding{193} naturality of $\cong$,
\ding{194} the fact that $(\withm A)^*$ is strongly generated and the definition of strength of $\withm A$,
\ding{195} properties of strength,
\ding{196} naturality of $\eta^M$ and $\emb$,
\ding{197} naturality of $\eta^M$,
\ding{198} properties of strength,
\ding{199} naturality of $\eta^M$.

It is left to show that the following diagram commutes:
\begin{equation*}
\begin{tikzpicture}
\node(a){$I$};
\node(b)[right of=a, xshift=4em]{$M(\withm A)^*I$};
\node(c)[below of=b]{$B$};
\path[arrows={-latex}, font=\scriptsize]
(a) edge node [auto] {$\eta^\res$} (b)
(a) edge node [below] {$u^B\ \ $} (c)
(b) edge node [auto] {$f$} (c)
;
\end{tikzpicture}
\end{equation*}
It commutes, since $\eta^\res = \mathfrak u$, and $f$ is a morphism between the monoid parts of $FA$ and $\tuple{B,b,m^B,u^B}$.
\end{proof}

\subsection{Theorem~\ref{thm:monadicity}}

We use the strict version of Beck's monadicity theorem~\cite{BeckPhD} (see Mac Lane~\cite[Sec.\ VI.7]{lane1998categories}). 
Let $h_0, h_1 : \tuple{A,a, m^A, u^A} \to \tuple{B,b,m^B,u^B}$ be a pair of morphisms between Eilenberg--Moore monoids. Assume that $g : B \to C$ is a split coequaliser of $Uh_0$ and $Uh_1$, that is, the following diagram commutes, and the top and bottom paths are both equal to $\id$:
\begin{equation*}
\begin{tikzpicture}
\node(a1){$B$};
\node(a2)[right of=a1, xshift=0em]{$A$};
\node(a3)[right of=a2, xshift=0em]{$B$};
\node(b1)[below of=a1]{$C$};
\node(b2)[right of=b1, xshift=0em]{$B$};
\node(b3)[right of=b2, xshift=0em]{$C$};
\path[arrows={-latex}, font=\scriptsize]
(a1) edge node [auto] {$t$} (a2)
(a2) edge node [auto] {$h_0$} (a3)
(b1) edge node [auto] {$s$} (b2)
(a1) edge node [left] {$g$} (b1)
(a2) edge node [auto] {$h_1$} (b2)
(b2) edge node [auto] {$g$} (b3)
(a3) edge node [auto] {$g$} (b3)
;
\end{tikzpicture}
\end{equation*}
By Beck's theorem, it is enough to show that there exists a unique Eilenberg--Moore monoid $\tuple{C, c, m^C, u^C}$ such that $g$ is a morphism of Eilenberg--Moore monoids $\tuple{B,b,m^B,u^B} \to \tuple{C,c,m^C,u^C}$ and a coequaliser of $h_0$ and $h_1$.

We notice that $h_0$ and $h_1$ are also morphisms of Eilenberg--Moore algebras $\tuple{A,a} \to \tuple{B,b}$. Hence, by the fact that the Eilenberg--Moore adjunction is monadic, there exists a unique Eilenberg--Moore algebra $\tuple{C,c}$ such that $g : \tuple{B,b} \to \tuple{C,c}$ is a unique coequaliser of $h_0$ and $h_1$ understood as morphisms between Eilenberg--Moore algebras. In detail, the algebra $\tuple{C,c}$ is equal to:
\begin{equation*}
\tuple{C ,\ \ MC \xrightarrow{Ms} MB \xrightarrow{b} B \xrightarrow{g} C}
\end{equation*}
Similarly, $h_0$ and $h_1$ are morphism between monoids $\tuple{A,m^A,u^A}$ and $\tuple{B,m^B,u^B}$. Since the adjunction between $\catvar C$ and the category of monoids is monadic, there exists a unique monoid $\tuple{C,m^C,u^C}$ such that $g : \tuple{B,m^B,u^B} \to \tuple{C,m^C,u^C}$ is a unique coequaliser of $h_0$ and $h_1$ understood as morphism between monoids. In detail, the monoid $\tuple{C,m^C,u^C}$ is given as:
\begin{equation*}
\tuple{C ,\ \ C\otimes C \xrightarrow{s\otimes s} B \otimes B \xrightarrow{m^B} B \xrightarrow{g} C ,\ \ I \xrightarrow{u^B} B \xrightarrow{g} C}
\end{equation*}

To check that $\tuple{C,c,m^C,u^C}$ is an Eilenberg--Moore monoid,  we need to verify the coherence condition:
\begin{equation*}
\begin{tikzpicture}
\node(a1){$MC \otimes C$};
\node(a2)[right of=a1, xshift=4em]{$MB \otimes C$};
\node(a4)[right of=a2, xshift=12em]{$B \otimes C$};
\node(a5)[right of=a4, xshift=4em]{$C \otimes C$};
\node(b1)[below of=a1]{$M(C \otimes C)$};
\node(b2)[right of=b1, xshift=4em]{$M(B \otimes C)$};
\node(b3)[right of=b2, xshift=4em]{$MB \otimes B$};
\node(b4)[right of=b3, xshift=4em]{$B \otimes B$};
\node(b5)[right of=b4, xshift=4em]{$B \otimes B$};
\node(c1)[below of=b1]{$M(B \otimes B)$};
\node(c4)[below of=b4]{$B$};
\node(c5)[below of=b5]{$B$};
\node(d1)[below of=c1]{$MB$};
\node(d2)[right of=d1, xshift=4em]{$MC$};
\node(d3)[right of=d2, xshift=4em]{$MB$};
\node(d4)[right of=d3, xshift=4em]{$B$};
\node(d5)[right of=d4, xshift=4em]{$C$};
\path[arrows={-latex}, font=\scriptsize]
(a1) edge node [auto] {$Ms \otimes \id$} (a2)
(a2) edge node [auto] {$b \otimes \id$} (a4)
(a4) edge node [auto] {$g \otimes \id$} (a5)
(b1) edge node [auto] {$M(s \otimes \id)$} (b2)
(b3) edge node [auto] {$b \otimes \id$} (b4)
(b4) edge node [right, yshift=-0.2em] {$g \otimes g$} (a5)
(a5) edge node [auto] {$s \otimes s$} (b5)
(a4) edge node [auto] {$\id \otimes s$} (b4)
(a2) edge node [auto, yshift=-0.3em] {$\id \otimes s$} (b3)
(a1) edge node [left] {$\tau$} (b1)
(a2) edge node [auto] {$\tau$} (b2)
(b1) edge node [left] {$M(s \otimes s)$} (c1)
(b2) edge node [left, yshift=0.8em, xshift=0.7em] {$M(\id \otimes s)$} (c1)
(b3) edge node [auto] {$\tau$} (c1)
(b4) edge node [auto] {$m^B$} (c4)
(c1) edge node [left] {$Mm^B$} (d1)
(b5) edge node [auto] {$m^B$} (c5)
(c5) edge node [auto] {$g$} (d5)
(d1) edge node [below] {$Mg$} (d2)
(d2) edge node [below] {$Ms$} (d3)
(d3) edge node [below] {$b$} (d4)
(d4) edge node [below] {$g$} (d5)
(d1) edge node [auto] {$b$} (c4)
(c4) edge node [auto] {$g$} (d5)
;
\node[right of=a1, xshift=0em, yshift=-1.5em] {\ding{192}};
\node[right of=a1, xshift=6em, yshift=-2.5em] {\ding{193}};
\node[right of=a1, xshift=13em, yshift=-1.5em] {\ding{194}};
\node[right of=a1, xshift=23em, yshift=-1.0em] {\ding{195}};
\node[right of=a1, xshift=-3em, yshift=-6.25em] {\ding{196}};
\node[right of=a1, xshift=13em, yshift=-6.5em] {\ding{197}};
\node[right of=a1, xshift=23em, yshift=-6.5em] {\ding{198}};
\node[right of=a1, xshift=20em, yshift=-10em] {\ding{199}};
\end{tikzpicture}
\end{equation*}
\ding{192} and
\ding{193} naturality of $\tau$,
\ding{194} $\otimes$ is a bifunctor,
\ding{195} $g \cdot s = \id$, since $g$ is a split coequaliser,
\ding{196} $\otimes$ is a bifunctor,
\ding{197} $\tuple{B,b,m^B,u^B}$ is an Eilenberg--Moore monoid,
\ding{198} $g$ is a morphism between monoids,
\ding{199} $g$ is a morphism between Eilenberg--Moore algebras.

Morphisms in $\catname{EMMon}_M$ need to be exactly morphisms between Eilenberg--Moore algebras and monoids. Thus, $g$ is a morphism $\tuple{B,b,m^B,u^B} \to \tuple{C,c,m^C,u^C}$ and a coequaliser of $h_0$ and $h_1$.

\subsection{Theorem~\ref{thm:mama}}

The triple $\tuple{\mexp,\ \comp,\ \ident}$ is a monoid, since it is a special case of the general construction $\tuple{\monoidalexp{A}{A},\ \comp,\ \ident}$; see, for example, Rivas \textit{et al.}~\cite{DBLP:conf/ppdp/RivasJS15}.
We need to verify that $p$ has the Eilenberg--Moore property.\footnote{Note that $\tuple{MA \Rightarrow MA, \mladj p}$ is defined as in Kock's~\cite{JAZ:4877532} construction of a closed monoidal structure on the Eilenberg--Moore category of a commutative monad. Kock's proof that $\tuple{MA \Rightarrow MA, \mladj p}$ is an Eilenberg--Moore algebra can be applied if we additionally assume that $\catvar C$ is symmetric.} First, we need to show that $\mladj{p} \cdot \eta = \id$. Using the naturality of $\mladj\spholder$, it is enough to show that $\mladj{p \cdot (\eta\otimes\id)} = \id$. Consider the following diagram, where the left-bottom path is equal to $p \cdot (\eta\otimes\id)$:
\begin{equation*}
\begin{tikzpicture}
\node(a1){$(\mexp) \otimes MA$};
\node(b1)[below of=a1,xshift=0em]{$M(\mexp) \otimes MA$};    
\node(b2)[right of=b1,xshift=10em]{$M((\mexp) \otimes MA)$};
\node(b3)[right of=b2,xshift=7em]{$MMA$};
\node(b4)[right of=b3,xshift=4em]{$MA$};
\path[arrows={-latex}, font=\scriptsize]
(a1) edge node [auto] {$\eta \otimes\id$} (b1)
(b1) edge node [auto] {$\tau$} (b2)
(b2) edge node [auto] {$M\app$} (b3)
(b3) edge node [auto] {$\mu$} (b4)
(a1) edge[bend left=15] node [auto] {$\eta$} (b2)
(a1) edge[bend left=15] node [auto] {$\app$} (b4)
;
\node[right of=a1, xshift=2em, yshift=-2em] {\ding{192}};
\node[right of=a1, xshift=12em, yshift=-1em] {\ding{193}};
\end{tikzpicture}
\end{equation*}
\ding{192} properties of strength,
\ding{193} naturality of $\eta$ and monad laws.

Thus, we obtain $\mladj{p \cdot (\eta\otimes\id)} = \mladj\app = \id$, since $\app$ is the counit of the adjunction.
To see that $\mladj p$ has the Eilenberg--Moore property, we also need to verify that $\mladj p \cdot M\mladj p = \mladj p \cdot \mu$. Using the naturality of $\mladj\spholder$, it is enough to verify that $\mladj{p \cdot (M\mladj{p}\otimes\id)} = \mladj{p \cdot (\mu \otimes \id)}$. Therefore, it is enough to show that $p \cdot (M\mladj{p}\otimes\id) = p \cdot (\mu \otimes \id)$. It is detailed in the following diagram:
\begin{equation*}
\begin{tikzpicture}
\node(c1){$MM(\mexp)\otimes MA$};
\node(d1)[right of=c1, xshift=26em]{$M(\mexp)\otimes MA$};
\node(c2)[below of=c1,xshift=18em]{$M(M(\mexp)\otimes MA)$};
\node(c3)[below of=c2]{$MM((\mexp)\otimes MA)$};
\node(c4)[below of=c3]{$MMMA$};
\node(d3)[below of=d1,yshift=-4em]{$M((\mexp)\otimes MA)$};
\node(d4)[below of=d3]{$MMA$};
\node(d5)[below of=d4]{$MA$};
\node(c5)[below of=c4]{$MMA$};
\node(b5)[below of=c1,yshift=-12em]{$M((\mexp)\otimes MA)$};
\node(b4)[above of=b5,yshift=4em]{$M(\mexp)\otimes MA$};
\path[arrows={-latex}, font=\scriptsize]
(c1) edge node [auto] {$\mu \otimes\id$} (d1)
(c1) edge node [auto] {$\tau$} (c2)
(c2) edge node [auto] {$M\tau$} (c3)
(c3) edge node [auto] {$MM\app$} (c4)
(d1) edge node [auto] {$\tau$} (d3)
(c3) edge node [auto] {$\mu$} (d3)
(d3) edge node [auto] {$M\app$} (d4)
(c4) edge node [auto] {$\mu$} (d4)
(d4) edge node [auto] {$\mu$} (d5)
(c5) edge node [auto] {$\mu$} (d5)
(c4) edge node [auto] {$M\mu$} (c5)
(c2.south west) edge[bend right=60] node [auto,yshift=-0.5em] {$Mp$} (c5.north west)
(c1) edge node [auto] {$M\mladj{p} \otimes \id$} (b4)
(b4) edge node [auto] {$\tau$} (b5)
(b5) edge node [auto] {$M\app$} (c5)
(c2.west) edge[bend right=0] node [auto,yshift=-2em,xshift=-2.3em] {$M(\mladj{p} \otimes \id)$} (b5)
;
\node[right of=c1, xshift=3em, yshift=-3.5em] {\ding{192}};
\node[right of=c1, xshift=22em, yshift=-3.5em] {\ding{193}};
\node[right of=c1, xshift=5em, yshift=-10.5em] {\ding{194}};
\node[right of=c1, xshift=11em, yshift=-10.5em] {\ding{195}};
\node[right of=c1, xshift=20em, yshift=-9.5em] {\ding{196}};
\node[right of=c1, xshift=20em, yshift=-13.5em] {\ding{197}};
\end{tikzpicture}
\end{equation*}
\ding{192} naturality of $\tau$,
\ding{193} properties of strength,
\ding{194} $\app$ is the counit of the adjunction,
\ding{195} definition of $p$,
\ding{196} naturality of $\mu$,
\ding{197} monad laws.

It is left to check that the coherence diagram commutes, which in this case instantiates as follows:
\begin{equation*}
\begin{tikzpicture}
\node(a) {$M(\mexp) \otimes (\mexp)$};
\node(b)[below of=a] {$M((\mexp) \otimes (\mexp))$};
\node(c)[right of=a, xshift=22em] {$(\mexp) \otimes (\mexp)$};
\node(d)[right of=b, xshift=10em] {$M(\mexp)$};
\node(e)[right of=d, xshift=8em] {$\mexp$};
\path[arrows={-latex}, font=\scriptsize]
(a) edge node [auto] {$\mladj{\mu \cdot M\app \cdot \tau} \otimes \id$} (c)
(a) edge node [auto] {$\tau$} (b)
(b) edge node [auto] {$M \comp$} (d)
(d) edge node [auto] {$\mladj{\mu \cdot M\app \cdot \tau}$} (e)
(c) edge node [auto] {$\comp$} (e)
;
\end{tikzpicture}
\end{equation*}
Since $\mradj{\spholder}$ is an isomorphism, it is enough to show that the $\mradj{\spholder}$-images of both paths in the diagram are equal. The top-right path:
\begin{align*}
\mradj{\comp \cdot (\mladj{\mu \cdot M\app \cdot \tau} \otimes \id)}
&=
\mradj{\comp} \cdot ((\mladj{\mu \cdot M\app \cdot \tau} \otimes \id) \otimes \id) & & \text{naturality of $\mradj{\spholder}$}
\\
&=
\mradj{\mladj k} \cdot ((\mladj{\mu \cdot M\app \cdot \tau} \otimes \id) \otimes \id) & & \text{definition of $\comp$}
\\
&=
k \cdot ((\mladj{\mu \cdot M\app \cdot \tau} \otimes \id) \otimes \id) & & \text{isomorphism}
\\
(*)
&=
k \cdot ((((\id \Rightarrow \mu) \cdot (\id \Rightarrow M\app) \cdot \mladj{\tau}) \otimes \id) \otimes \id) & & \text{naturality of $\mladj{\spholder}$}
\end{align*}
The left-bottom path:
\begin{align*}
\mradj{\mladj{\mu \cdot M\app \cdot \tau} \cdot M\app \cdot \tau}
&=
\mradj{\mladj{\mu \cdot M\app \cdot \tau}} \cdot (M\app \otimes \id) \cdot (\tau \otimes \id) & & \text{naturality of $\mradj{\spholder}$}
\\
(**)
&=
\mu \cdot M\app \cdot \tau \cdot (M\app \otimes \id) \cdot (\tau \otimes \id) & & \text{isomorphism}
\end{align*}
To show that $(*)$ is equal to $(**)$, we split the desired diagram into two. First, consider the following diagram, where the left-bottom path is equal to~$(*)$. For brevity, we denote the object $MA\Rightarrow MA$ as $E$.
\begin{equation*}
\begin{tikzpicture}
\node(a) {$(ME \otimes E)\otimes MA$};
\node(b)[below of=a, yshift=-4em] {$((MA \Rightarrow M(E \otimes MA)) \otimes E) \otimes MA$};
\node(c)[below of=b] {$((MA \Rightarrow MMA) \otimes E) \otimes MA$};
\node(d)[below of=c] {$(E \otimes E) \otimes MA$};
\node(e)[below of=d] {$E \otimes (E \otimes MA)$};
\node(f)[right of=e, xshift=8em] {$E \otimes MA$};
\node(g)[right of=f, xshift=16em] {$MA$};
\node(h1)[right of=a, xshift=12em,yshift=0em] {$ME\otimes(E\otimes MA)$};
\node(h2)[below of=h1, xshift=-8em] {$(MA \Rightarrow M(E \otimes MA)) \otimes (E \otimes MA)$};
\node(h3)[below of=h1,yshift=-4em] {$(MA \Rightarrow M(E \otimes MA)) \otimes MA$};
\node(h4)[below of=h3] {$(MA \Rightarrow MMA) \otimes MA$};
\node(h5)[below of=h4] {$E \otimes MA$};
\node(hh1)[right of=h1, xshift=12em] {$ME\otimes MA$};
\node(hh2)[right of=h2, xshift=12em] {$(MA \Rightarrow M(E \otimes MA)) \otimes MA$};
\node(j)[above of=g] {$MMA$};
\node(k)[above of=j] {$M(E \otimes MA)$};
\path[arrows={-latex}, font=\scriptsize]
(a) edge node [auto, yshift=2em] {$(\mladj{\tau} \otimes \id)\otimes \id$} (b)
(b) edge node [auto] {$((\id \Rightarrow M\app) \otimes \id) \otimes \id$} (c)
(c) edge node [auto] {$((\id \Rightarrow \mu) \otimes \id) \otimes \id$} (d)
(d) edge node [auto] {$\cong$} (e)
(e) edge node [auto] {$\id \otimes \app$} (f)
(f) edge node [auto] {$\app$} (g)
(h1) edge node [left, xshift=-0.5em] {$\mladj{\tau} \otimes\id$} (h2)
(h3) edge node [auto] {$(\id \Rightarrow M\app) \otimes \id$} (h4)
(h4) edge node [auto] {$(\id \Rightarrow \mu) \otimes \id$} (h5)
(a) edge node [auto] {$\cong$} (h1)
(b) edge node [auto] {$\cong$} (h2)
(b) edge node [auto] {$\comp \otimes \id$} (h3)
(c) edge node [auto] {$\comp \otimes \id$} (h4)
(d) edge node [auto] {$\comp \otimes \id$} (h5)
(h5) edge node [auto] {$\app$} (g)
(hh1) edge node [auto] {$\mladj{\tau} \otimes \id$} (hh2)
(h1) edge node [auto] {$\id \otimes \app$} (hh1)
(h2) edge node [auto] {$\id \otimes \app$} (hh2)
(j) edge node [auto] {$\mu$} (g)
(k) edge node [auto] {$M\app$} (j)
(h3.east) edge node [auto] {$\app$} (k)
(hh2) edge node [auto] {$\app$} (k)
(hh1) edge node [auto] {$\tau$} (k)
;
\node[below of=a, xshift=6.75em, yshift=2em] {\ding{192}};
\node[below of=a, xshift=20em, yshift=2em] {\ding{193}};
\node[below of=a, xshift=15em, yshift=-2em] {\ding{194}};
\node[below of=a, xshift=29em, yshift=-2em] {\ding{195}};
\node[below of=a, xshift=12em, yshift=-6em] {\ding{196}};
\node[below of=a, xshift=12em, yshift=-10em] {\ding{197}};
\node[below of=a, xshift=25em, yshift=-10em] {\ding{198}};
\node[below of=a, xshift=12em, yshift=-14em] {\ding{199}};
\end{tikzpicture}
\end{equation*}
\ding{192} naturality of $\cong$,
\ding{193} $\otimes$ is a bifunctor,
\ding{194} and
\ding{195} $\app$ is the counit of the adjunction,
\ding{196} and
\ding{197}~naturality of $\comp$,
\ding{198} naturality of $\app$,
\ding{199} $\app$ is the counit of the adjunction.

Now, consider the following diagram, where the left-most path is equal to the top-right path of the diagram above, while the path around the perimeter is equal to $(**)$:
\begin{equation*}
\begin{tikzpicture}
\node(a) {$(ME \otimes E)\otimes MA$};
\node(b)[below of=a] {$ME \otimes (E\otimes MA)$};
\node(c)[below of=b] {$ME \otimes MA$};
\node(d)[below of=c] {$M(E \otimes MA)$};
\node(e)[below of=d] {$MMA$};
\node(f)[below of=e] {$MA$};
\node(cc)[right of=c, xshift=6em] {$M(E \otimes (E\otimes MA))$};
\node(ccc)[right of=cc, xshift=6em] {$M((E \otimes E)\otimes MA)$};
\node(x)[right of=a, xshift=25em] {$M(E\otimes E)\otimes MA$};
\node(y)[right of=c, xshift=25em] {$ME\otimes MA$};
\node(z)[right of=e, xshift=25em] {$M(E\otimes MA)$};
\path[arrows={-latex}, font=\scriptsize]
(a) edge node [auto] {$\cong$} (b)
(b) edge node [auto] {$\id \otimes \app$} (c)
(c) edge node [auto] {$\tau$} (d)
(d) edge node [auto] {$M\app$} (e)
(e) edge node [auto] {$\mu$} (f)
(b) edge node [auto] {$\tau$} (cc)
(cc) edge node [auto] {$M(\id \otimes \app)$} (d)
(ccc) edge node [above] {$M\cong$} (cc)
(x) edge node [auto] {$M\comp \otimes \id$} (y)
(y) edge node [auto] {$\tau$} (z)
(a) edge node [auto] {$\tau \otimes \id$} (x)
(z) edge node [above] {$M\app$} (e)
(x) edge node [auto] {$\tau$} (ccc)
(ccc) edge node [auto,yshift=0.5em, xshift=-1em] {$M(\comp \otimes \id)$} (z)
;
\node[below of=a, xshift=15em, yshift=0em] {\ding{192}};
\node[below of=a, xshift=4em, yshift=-4em] {\ding{193}};
\node[below of=a, xshift=25.25em, yshift=-4em] {\ding{194}};
\node[below of=a, xshift=15em, yshift=-8em] {\ding{195}};
\end{tikzpicture}
\end{equation*}
\ding{192} properties of strength,
\ding{193} and \ding{194} naturality of strength,
\ding{195} $\app$ is the counit of the adjunction.

\subsection{Theorem~\ref{thm:rep}}

The fact that $\mladj m$ is a morphism between monoids follows from the construction of Cayley representation of monoids in monoidal categories~\cite{RivasJaskelioff2014,DBLP:conf/ppdp/RivasJS15}. The retraction $r : (MA \Rightarrow MA) \to MA$ is defined as follows:
\begin{equation*}
r = \Bigl(
(MA \Rightarrow MA)
\xrightarrow{\cong}
(MA \Rightarrow MA) \otimes I
\xrightarrow{\id \otimes u}
(MA \Rightarrow MA) \otimes MA
\xrightarrow{\app}
MA
\Bigr)
\end{equation*}

To check that $\mladj{m}$ is a morphism between Eilenberg--Moore parts, we want the following diagram to commute:
\begin{equation*}
\begin{tikzpicture}
\node(a0){$MMA$};    
\node(a1)[right of=a0, xshift=4em]{$M(MA \Rightarrow MA)$};
\node(b0)[below of=a0, xshift=0em]{$MA$};
\node(b1)[below of=a1, xshift=0em]{$MA \Rightarrow MA$};
\path[arrows={-latex}, font=\scriptsize]
(a0) edge node [auto] {$M \mladj m$} (a1)
(b0) edge node [auto] {$\mladj m$} (b1)
(a0) edge node [auto] {$\mu$} (b0)
(a1) edge node [auto] {$\mladj p$} (b1)
;
\end{tikzpicture}
\end{equation*}

Using the naturality of $\mladj\spholder$, it is enough to show that $\mladj{p \cdot (M\mladj{m} \otimes \id)} = \mladj{m \cdot (\mu \otimes \id)}$. Thus, it is enough to show that $p \cdot (M\mladj{m} \otimes \id) = m \cdot (\mu \otimes \id)$. It is detailed in the following diagram, where the right-most edge is equal to $p$:
\begin{equation*}
\begin{tikzpicture}
\node(a0){$MMA \otimes MA$};    
\node(a1)[right of=a0, xshift=16em]{$M(MA \Rightarrow MA) \otimes MA$};
\node(b0)[below of=a0, yshift=-8em]{$MA \otimes MA$};
\node(b1)[below of=a1, xshift=0em]{$M((MA \Rightarrow MA) \otimes MA)$};
\node(c1)[below of=b1, xshift=0em]{$MMA$};
\node(d1)[below of=c1, xshift=0em]{$MA$};
\node(x)[left of=b1, xshift=-10em]{$M(MA\otimes MA)$};
\path[arrows={-latex}, font=\scriptsize]
(a0) edge node [auto] {$M \mladj m \otimes \id$} (a1)
(b0) edge node [auto] {$m$} (d1)
(a0) edge node [auto] {$\mu \otimes \id$} (b0)
(a1) edge node [auto] {$\tau$} (b1)
(b1) edge node [auto] {$M\app$} (c1)
(c1) edge node [auto] {$\mu$} (d1)
(a0) edge node [auto] {$\tau$} (x)
(x) edge node [auto] {$M(\mladj m \otimes \id)$} (b1)
(x) edge node [auto] {$Mm$} (c1)
;
\node[below of=a, xshift=10em, yshift=2.5em] {\ding{192}};
\node[below of=a, xshift=18em, yshift=-1.5em] {\ding{193}};
\node[below of=a, xshift=6em, yshift=-4em] {\ding{194}};
\end{tikzpicture}
\end{equation*}
\ding{192} naturality of $\tau$,
\ding{193} $\app$ is the counit of the adjunction,
\ding{194} coherence for $\tuple{MA,\mu,m,u}$.

To verify that $\mladj m$ preserves the unit, we want the following diagram:
\begin{equation*}
\begin{tikzpicture}
\node(a){$I$};
\node(b)[right of=a, xshift=2em]{$MA$};
\node(c)[below of=b]{$MA\Rightarrow MA$};
\path[arrows={-latex}, font=\scriptsize]
(a) edge node [auto] {$u$} (b)
(a) edge node [auto,xshift=-0.2em] {$\ident$} (c)
(b) edge node [auto] {$\mladj m$} (c)
;
\end{tikzpicture}
\end{equation*}
Using the naturality of $\mladj\spholder$, it is enough to show that $\mladj{m \cdot (u \otimes \id)} = \ident$. Thus, it is enough to show that $m \cdot (u \otimes \id) = \mradj{\ident} = (I \otimes MA \xrightarrow{\cong} MA)$. This follows from the fact that $\tuple{MA, m, u}$ is a monoid.

\subsection{Theorem~\ref{thm:semmonisemmon}}

Except for the assumption that $M$ is commutative, the only difference between the definition of symmetric Eilenberg--Moore monoids and the definition of regular Eilenberg--Moore monoids is in the coherence conditions. Thus, given a symmetric Eilenberg--Moore monoid $\tuple{A,a,m,u}$, it is enough to show that it satisfies the coherence condition for Eilenberg--Moore monoids:
\begin{equation*}
\begin{tikzpicture}
\node(a0){$MA \otimes A$};
\node(a1)[below of=a0, yshift=-12em]{$M(A \otimes A)$};
\node(b0)[below of=a0, xshift=6em]{$MA \otimes MA$};
\node(b1)[below of=b0, xshift=0em]{$M(MA \otimes A)$};
\node(b2)[below of=b1, xshift=0em]{$MM(A \otimes A)$};
\node(c0)[right of=b0, xshift=4em]{$M(A \otimes MA)$};
\node(c1)[right of=b1, xshift=4em]{$MM(A \otimes A)$};
\node(c2)[below of=c1, xshift=0em]{$M(A \otimes A)$};
\node(c3)[below of=c2, xshift=0em]{$MA$};
\node(xx)[right of=c1, xshift=2em, yshift=-4em]{$MA$};
\node(yy)[right of=c3, xshift=8em, yshift=0em]{$A$};
\node(zz)[right of=a0, xshift=22em, yshift=0em]{$A \otimes A$};
\path[arrows={-latex}, font=\scriptsize]
(a0) edge node [auto] {$\tau$} (a1)
(a0) edge node [auto] {$\id \otimes \eta$} (b0)
(b0) edge node [auto] {$\tau'$} (b1)
(b1) edge node [auto] {$M\tau$} (b2)
(a1) edge node [auto] {$\eta$} (b2)
(a0) edge node [below, yshift=-0.3em] {$\eta$} (b1)
(b0) edge node [below] {$\tau$} (c0)
(c0) edge node [auto] {$M\tau'$} (c1)
(c1) edge node [auto] {$\mu$} (c2)
(b2) edge node [auto] {$\mu$} (c2)
(c2) edge node [auto] {$Mm$} (xx)
(a1) edge node [auto] {$Mm$} (c3)
(c3) edge node [auto] {$a$} (yy)
(zz) edge node [auto] {$m$} (yy)
(a0) edge node [auto] {$a \otimes \id$} (zz)
(xx) edge node [auto] {$a$} (yy)
(a1) edge node [auto] {$\id$} (c2)
(b0) edge node [auto] {$a \otimes a$} (zz)
;
\node[below of=a0, xshift=1.25em, yshift=0em] {\ding{192}};
\node[below of=a0, xshift=3.2em, yshift=1em] {\ding{193}};
\node[below of=a0, xshift=7em, yshift=2em] {\ding{194}};
\node[below of=a0, xshift=10em, yshift=-4em] {\ding{195}};
\node[below of=a0, xshift=20em, yshift=-4em] {\ding{196}};
\node[below of=a0, xshift=8em, yshift=-9em] {\ding{197}};
\node[below of=a0, xshift=14em, yshift=-10em] {\ding{198}};
\end{tikzpicture}
\end{equation*}
\ding{192} naturality of $\eta$,
\ding{193} properties of strength,
\ding{194} $a$ has the Eilenberg--Moore property,
\ding{195} $M$ is commutative,
\ding{196} coherence for symmetric Eilenberg--Moore monoids,
\ding{197} monad laws,
\ding{198} identity.

\subsection{Theorem~\ref{thm:revisitingdefofendo}}

The fact that $\mladj{q}$ has the Eilenberg--Moore property follows from exactly the same reasoning as for $\mladj{p}$ in the proof of Theorem~\ref{thm:mama}.

It is left to verify the coherence condition. It is shown by the following diagram (and subsequent two diagrams, each detailing a part of the previous one), where $E$ stands for $A \Rightarrow MA$:
\begin{equation*}
\begin{tikzpicture}
\node(a0){$(ME\otimes ME)\otimes A$};
\node(ax)[right of=a0, xshift=28em]{$(E \otimes E) \otimes A$};
\node(bx)[below of=ax]{$E \otimes (E\otimes A)$};
\node(cx)[below of=bx,yshift=-4em]{$E \otimes MA$};
\node(dx)[below of=cx]{$M(E\otimes A)$};
\node(ex)[below of=dx]{$MMA$};
\node(fx)[below of=ex]{$MA$};
\node(by)[left of=bx, xshift=-8em]{$ME \otimes (ME\otimes A)$};
\node(cy)[left of=cx, xshift=-8em,yshift=4em]{$ME \otimes M(E\otimes A)$};
\node(dy)[left of=dx, xshift=-8em,yshift=4em]{$ME \otimes MMA$};
\node(b0)[below of=a0]{$M(E \otimes ME) \otimes A$};
\node(c0)[below of=b0]{$MM(E\otimes E)\otimes A$};
\node(d0)[below of=c0, yshift=-4em]{$MME\otimes A$};
\node(e0)[below of=d0, yshift=-4em]{$ME \otimes A$};
\node(e1)[right of=e0, xshift=4em]{$M(E \otimes A)$};
\node(e2)[right of=e1, xshift=4em]{$MMA$};
\node(c1)[right of=c0, xshift=4em]{$M(E\otimes E)\otimes A$};
\node(d1)[right of=d0, xshift=4em]{$M((E\otimes E)\otimes A)$};
\node(d2)[right of=d1, xshift=4em]{$M(E \otimes (E\otimes A))$};
\node(d3)[right of=d2, xshift=6em]{$M(E\otimes MA)$};
\node(d4)[below of=d3, xshift=0em]{$MM(E\otimes A)$};
\node(d5)[left of=d4, xshift=-6em]{$MMMA$};
\path[arrows={-latex}, font=\scriptsize]
(a0) edge node [auto] {$(\mladj q \otimes \mladj q) \otimes \id$} (ax)
(by) edge node [auto] {$\mladj q \otimes (\mladj q \otimes \id)$} (bx)
(ax) edge node [auto] {$\cong$} (bx)
(a0) edge node [auto] {$\cong$} (by)
(bx) edge node [auto] {$\id \otimes \app$} (cx)
(cx) edge node [auto] {$\tau'$} (dx)
(dx) edge node [auto] {$M\app$} (ex)
(ex) edge node [auto] {$\mu$} (fx)
(a0) edge node [auto] {$\tau \otimes \id$} (b0)
(cy) edge node [auto] {$\id \otimes M\app$} (dy)
(by) edge node [auto, yshift=-0.2em] {$\mladj q \otimes q$} (cx)
(by) edge node [auto] {$\id \otimes \tau$} (cy)
(dy) edge node [auto, xshift=1em] {$\mladj q \otimes \mu$} (cx)
(b0) edge node [auto] {$M\tau' \otimes \id$} (c0)
(c0) edge node [auto] {$MM \kcomp \otimes \id$} (d0)
(d0) edge node [auto] {$\mu \otimes \id$} (e0)
(e0) edge node [auto] {$\tau$} (e1)
(e1) edge node [auto] {$M\app$} (e2)
(e2) edge node [auto] {$\mu$} (fx)
(c0) edge node [auto] {$\mu \otimes \id$} (c1)
(c1) edge node [auto] {$\tau$} (d1)
(d1) edge node [auto] {$M(\kcomp \otimes \id)$} (e1)
(d1) edge node [auto] {$M\cong$} (d2)
(d2) edge node [auto] {$M(\id \otimes \app)$} (d3)
(d3) edge node [auto] {$M\tau'$} (d4)
(d5) edge node [auto] {$M\mu$} (e2)
(d4) edge node [above] {$MM\app$} (d5)
;
\node[below of=a0, xshift=6em, yshift=0em] {\ding{192}};
\node[below of=a0, xshift=18em, yshift=2.5em] {\ding{193}};
\node[below of=a0, xshift=26em, yshift=-1.5em] {\ding{194}};
\node[below of=a0, xshift=25em, yshift=-5.5em] {\ding{195}};
\node[below of=a0, xshift=3.7em, yshift=-10em] {\ding{196}};
\node[below of=a0, xshift=12em, yshift=-14em] {\ding{197}};
\end{tikzpicture}
\end{equation*}
\ding{192} detailed below,
\ding{193} naturality of $\cong$,
\ding{194} $\app$ is the counit of the adjunction,
\ding{195} definition of $q$,
\ding{196} naturality of $\mu$ and $\tau$,
\ding{197} $\app$ is the counit of the adjunction. TUTAJ!!

\begin{equation*}
\begin{tikzpicture}
\node(a0){$(ME \otimes ME) \otimes A$};
\node(a1)[right of=a0, xshift=4.5em]{$ME \otimes (ME \otimes A)$};
\node(a2)[right of=a1, xshift=5em]{$ME \otimes M(E \otimes A)$};
\node(a3)[right of=a2, xshift=14em]{$ME \otimes MMA$};
\node(b0)[below of=a0, yshift=-4em]{$M(E \otimes ME) \otimes A$};
\node(b1)[below of=a1, yshift=-4em]{$M((E \otimes ME) \otimes A)$};
\node(b2)[above of=b1, yshift=0em]{$M(E \otimes (ME \otimes A))$};
\node(b3)[below of=a2, yshift=0em]{$M(E \otimes M(E \otimes A))$};
\node(x1)[below of=b3, yshift=0em]{$M(E \otimes MMA)$};
\node(z1)[below of=a3, yshift=0em]{$E \otimes MA$};
\node(y1)[left of=z1, xshift=-4em]{$ME \otimes MA$};
\node(y2)[below of=y1]{$M(ME \otimes A)$};
\node(z2)[below of=z1]{$M(E \otimes A)$};
\node(z3)[below of=z2]{$MMA$};
\node(z4)[below of=z3, yshift=-8em]{$MA$};
\node(x2)[below of=x1]{$M(E\otimes MA)$};
\node(x3)[below of=x2]{$MM(E\otimes A)$};
\node(y3)[below of=y2]{$MM(E\otimes A)$};
\node(x4)[below of=x3]{$M(E\otimes A)$};
\node(y4)[below of=y3]{$MMMA$};
\node(y5)[below of=y4]{$MMA$};
\node(c0)[below of=b0]{$MM(E \otimes E) \otimes A$};
\node(d0)[below of=c0]{$M(E \otimes E) \otimes A$};
\node(e0)[below of=d0]{$M((E \otimes E) \otimes A)$};
\node(f0)[below of=e0]{$M(E \otimes (E \otimes A))$};
\node(f1)[right of=f0,xshift=4.5em]{$M(E \otimes MA)$};
\node(f2)[right of=f1,xshift=3.5em]{$MM(E \otimes A)$};
\node(f3)[right of=f2,xshift=2.5em]{$MMMA$};
\node(f4)[right of=f3,xshift=1em]{$MMA$};
\node(c1)[right of=c0,xshift=4.5em]{$M(M(E \otimes E) \otimes A)$};
\node(d1)[below of=c1]{$MM((E\otimes E)\otimes A)$};
\node(e1)[below of=d1]{$MM(E\otimes (E\otimes A))$};
\path[arrows={-latex}, font=\scriptsize]
(a0) edge node [auto, xshift=-0.1em] {$\cong$} (a1)
(a1) edge node [auto] {$\id \otimes \tau$} (a2)
(a2) edge node [auto] {$\id \otimes M\app$} (a3)
(a0) edge node [auto] {$\tau \otimes \id$} (b0)
(b0) edge node [auto, xshift=-0.1em] {$\tau$} (b1)
(b1) edge node [auto] {$M\cong$} (b2)
(a1) edge node [auto] {$\tau$} (b2)
(a2) edge node [auto] {$\tau$} (b3)
(b2) edge node [auto,yshift=0.2em] {$M(\id \otimes \tau)$} (b3)
(b3) edge node [auto,yshift=0.2em] {$M(\id \otimes M\app)$} (x1)
(a3) edge node [left,yshift=-0.5em] {$\mladj q \otimes \mu$} (z1)
(a3) edge node [auto,xshift=-2.4em,yshift=1em] {$\id \otimes \mu$} (y1)
(y1) edge node [auto] {$\tau'$} (y2)
(z1) edge node [auto] {$\tau'$} (z2)
(z2) edge node [auto] {$M\app$} (z3)
(z3) edge node [auto] {$\mu$} (z4)
(y2) edge node [auto,yshift=0.2em] {$M(\mladj q \otimes \id)$} (z2)
(y2) edge node [auto,yshift=-0.4em] {$Mq$} (z3)
(x1) edge node [auto] {$M(\id \otimes \mu)$} (x2)
(y1.west) edge node [left] {$\tau$} (x2.east)
(x2) edge node [auto] {$M\tau'$} (x3)
(y2) edge node [auto] {$M\tau$} (y3)
(y3) edge node [auto] {$MM\app$} (y4)
(y4) edge node [auto,xshift=0.7em] {$M\mu$} (z3)
(x3) edge node [auto] {$\mu$} (x4)
(y3.west) edge node [left] {$\mu$} (x4.north east)
(y5) edge node [auto] {$\mu$} (z4)
(y4) edge node [auto] {$\mu$} (y5)
(x4) edge node [auto] {$M\app$} (y5)
(b0) edge node [auto] {$M\tau' \otimes \id$} (c0)
(c0) edge node [auto] {$\mu \otimes \id$} (d0)
(d0) edge node [auto] {$\tau$} (e0)
(e0) edge node [auto] {$M\cong$} (f0)
(f0) edge node [auto,yshift=0.2em] {$M(\id \otimes \app)$} (f1)
(f1) edge node [auto,yshift=0.0em] {$M\tau'$} (f2)
(f2) edge node [auto,yshift=0.2em] {$MM\app$} (f3)
(f3) edge node [auto,yshift=0.0em] {$M\mu$} (f4)
(f4) edge node [auto,yshift=0.0em] {$\mu$} (z4)
(c0) edge node [auto, xshift=-0.1em] {$\tau$} (c1)
(b1) edge node [auto] {$M(\tau' \otimes \id)$} (c1)
(c1) edge node [auto, xshift=-0.1em] {$M\tau$} (d1)
(d1.south west) edge node [above, xshift=0.1em] {$\mu$} (e0)
(d1) edge node [auto] {$MM\cong$} (e1)
(e1.south west) edge node [auto,yshift=0.5em] {$\mu$} (f0)
;
\node[below of=a0, xshift=4em, yshift=2em] {\ding{192}};
\node[below of=a0, xshift=13em, yshift=2.5em] {\ding{193}};
\node[below of=a0, xshift=24em, yshift=2em] {\ding{194}};
\node[below of=a0, xshift=4em, yshift=-6em] {\ding{195}};
\node[below of=a0, xshift=14em, yshift=-6em] {\ding{196}};
\node[below of=a0, xshift=22em, yshift=-8em] {\ding{197}};
\node[below of=a0, xshift=32em, yshift=-1em] {\ding{198}};
\node[below of=a0, xshift=34.5em, yshift=-6em] {\ding{199}};
\node[below of=a0, xshift=32em, yshift=-8em] {\ding{200}};
\node[below of=a0, xshift=4em, yshift=-10em] {\ding{201}};
\node[below of=a0, xshift=24em, yshift=-12em] {\ding{203}};
\node[below of=a0, xshift=32em, yshift=-12em] {\ding{204}};
\node[below of=a0, xshift=6em, yshift=-14em] {\ding{202}};
\end{tikzpicture}
\end{equation*}
\ding{192} properties of strength,
\ding{193} 
\ding{194} 
\ding{195} naturality of $\tau$,
\ding{196} detailed below,
\ding{197} $M$ is commutative,
\ding{198} naturality of $\tau'$,
\ding{199} $\app$ is the counit of the adjunction,
\ding{200} definition of $q$,
\ding{201} properties of strength,
\ding{202} and
\ding{203} naturality of $\mu$,
\ding{204} monad laws.
\begin{equation*}
\begin{tikzpicture}
\node(a0){$M((E \otimes ME) \otimes A)$};
\node(a1)[right of=a0, xshift=6em]{$M(E \otimes (ME \otimes A))$};
\node(a2)[right of=a1, xshift=6em]{$M(E \otimes M(E \otimes A))$};
\node(a3)[right of=a2, xshift=7em]{$M(E \otimes MMA)$};
\node(b0)[below of=a0,yshift=0.5em]{$M(M(E \otimes E) \otimes A)$};
\node(b1)[below of=b0,yshift=0.5em]{$MM((E \otimes E) \otimes A)$};
\node(b2)[below of=b1,yshift=0.5em]{$MM(E \otimes (E \otimes A))$};
\node(b3)[below of=a3,yshift=0.5em]{$M(E \otimes MA)$};
\node(c3)[below of=b3,yshift=0.5em]{$MM(E \otimes A)$};
\node(d3)[below of=c3,yshift=0.5em]{$M(E\otimes A)$};
\node(e3)[below of=d3,yshift=0.5em]{$MMA$};
\node(g3)[below of=e3,yshift=0.5em]{$MA$};
\node(x0)[below of=b2,yshift=0.5em]{$M(E \otimes (E \otimes A))$};
\node(y0)[right of=b2, xshift=4.5em]{$MM(E \otimes MA)$};
\node(z0)[right of=y0, xshift=4.5em]{$MMM(E \otimes A)$};
\node(zz)[right of=z0, xshift=4em]{$MM(E \otimes A)$};
\node(x1)[below of=x0,yshift=0.5em]{$M(E \otimes MA)$};
\node(f2)[right of=x1,xshift=3.5em]{$MM(E \otimes A)$};
\node(f3)[right of=f2,xshift=2.5em]{$MMMA$};
\node(f4)[right of=f3,xshift=1em]{$MMA$};
\node(pp)[above of=f2,xshift=12em,yshift=-0.5em]{$MMMMA$};
\node(qq)[right of=pp,xshift=2em]{$MMMA$};
\path[arrows={-latex}, font=\scriptsize]
(a0) edge node [auto] {$M(\tau' \otimes \id)$} (b0)
(a0) edge node [auto,xshift=-0.1em] {$M\cong$} (a1)
(a1) edge node [auto,yshift=0.2em] {$M(\id \otimes \tau)$} (a2)
(a2) edge node [auto] {$M\tau'$} (b2)
(b0) edge node [auto] {$M\tau$} (b1)
(b1) edge node [auto] {$MM\cong$} (b2)
(a2) edge node [auto,yshift=0.2em] {$M(\id \otimes M\app)$} (a3)
(a3) edge node [auto] {$M(\id \otimes \mu)$} (b3)
(b2) edge node [auto] {$\mu$} (x0)
(x0) edge node [left] {$M(\id\otimes\app)$} (x1)
(x1) edge node [auto,yshift=0.0em] {$M\tau'$} (f2)
(f2) edge node [auto,yshift=0.2em] {$MM\app$} (f3)
(f3) edge node [auto,yshift=0.0em] {$M\mu$} (f4)
(f4) edge node [auto,yshift=0.0em] {$\mu$} (g3)
(b3) edge node [auto] {$M\tau'$} (c3)
(c3) edge node [auto] {$\mu$} (d3)
(d3) edge node [auto] {$M\app$} (e3)
(e3) edge node [auto] {$\mu$} (g3)
(b2) edge node [auto,yshift=0.2em] {$MM(\id \otimes \app)$} (y0)
(a3) edge node [auto] {$M\tau'$} (y0)
(y0) edge node [auto] {$\mu$} (x1)
(y0) edge node [auto] {$MM\tau'$} (z0)
(z0.north east) edge node [auto] {$M\mu$} (c3.west)
(z0) edge node [auto] {$\mu$} (f2)
(z0) edge node [auto,yshift=-0.3em] {$MMM\app$} (pp)
(pp) edge node [left] {$\mu$} (f3)
(z0) edge node [auto] {$\mu$} (zz)
(zz) edge node [auto] {$\mu$} (d3)
(pp) edge node [auto] {$\mu$} (qq)
(qq) edge node [auto] {$\mu$} (e3)
;
\node[below of=a0, xshift=8em, yshift=1em] {\ding{192}};
\node[below of=a0, xshift=18em, yshift=1em] {\ding{193}};
\node[below of=a0, xshift=26em, yshift=-1em] {\ding{194}};
\node[below of=a0, xshift=26em, yshift=-5em] {\ding{195}};
\node[below of=a0, xshift=3em, yshift=-8.25em] {\ding{196}};
\node[below of=a0, xshift=8em, yshift=-10em] {\ding{197}};
\node[below of=a0, xshift=15.5em, yshift=-10em] {\ding{198}};
\node[below of=a0, xshift=26em, yshift=-8.25em] {\ding{199}};
\node[below of=a0, xshift=26em, yshift=-11.5em] {\ding{200}};
\end{tikzpicture}
\end{equation*}
\ding{192} $M$ is a bistrong monad,
\ding{193} naturality of $\tau'$,
\ding{194} properties of strength,
\ding{195} monad laws,
\ding{196}
\ding{197}
\ding{198} 
\ding{199} naturality of $\mu$,
\ding{200} monad laws.

\subsection{Theorem~\ref{thm:revisitingrepresentation}}

We need to show that $\mladj s$ is a morphism between Eilenberg--Moore algebras. It is enough to show that the following diagram commutes:
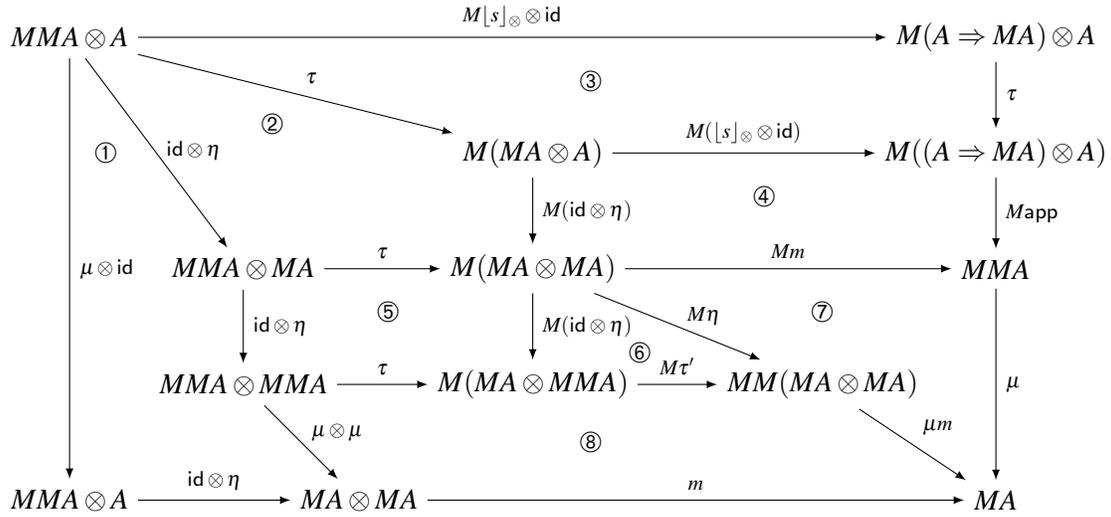
\begin{equation*}
\begin{tikzpicture}
\node(a0){$MMA \otimes A$};
\node(a1)[below of=a0, yshift=-12em]{$MMA \otimes A$};
\node(b0)[below of=a0,xshift=6em,yshift=-4em]{$MMA \otimes MA$};
\node(c0)[below of=a0,xshift=16em,yshift=0em]{$M(MA \otimes A)$};
\node(b1)[below of=b0,xshift=0em,yshift=0em]{$MMA \otimes MMA$};
\node(x0)[right of=a0, xshift=28em]{$M(A \Rightarrow MA) \otimes A$};
\node(x1)[below of=x0, xshift=0em]{$M((A \Rightarrow MA) \otimes A)$};
\node(x2)[below of=x1, xshift=0em]{$MMA$};
\node(x3)[below of=x2, yshift=-4em]{$MA$};
\node(c1)[below of=c0, xshift=0em]{$M(MA \otimes MA)$};
\node(c2)[below of=c1, xshift=0em]{$M(MA \otimes MMA)$};
\node(c3)[below of=c2, xshift=-6em]{$MA \otimes MA$};
\node(y0)[right of=c2, xshift=6em]{$MM(MA \otimes MA)$};
\path[arrows={-latex}, font=\scriptsize]
(a0) edge node [auto] {$\mu \otimes \id$} (a1)
(a0) edge node [auto] {$\tau$} (c0)
(a0) edge node [auto,yshift=-0.6em] {$\id \otimes \eta$} (b0)
(b0) edge node [auto] {$\id \otimes \eta$} (b1)
(a0) edge node [auto] {$M \mladj s \otimes \id$} (x0)
(x0) edge node [auto] {$\tau$} (x1)
(x1) edge node [auto] {$M\app$} (x2)
(c0) edge node [auto] {$M(\mladj s \otimes \id)$} (x1)
(b0) edge node [auto] {$\tau$} (c1)
(c0) edge node [auto] {$M(\id \otimes \eta)$} (c1)
(c1) edge node [auto] {$Mm$} (x2)
(c1) edge node [auto] {$M(\id \otimes \eta)$} (c2)
(b1) edge node [auto] {$\tau$} (c2)
(x2) edge node [auto] {$\mu$} (x3)
(a1) edge node [auto] {$\id \otimes \eta$} (c3)
(c3) edge node [auto] {$m$} (x3)
(b1) edge node [auto,yshift=-0.3em] {$\mu \otimes \mu$} (c3)
(c2) edge node [auto] {$M\tau'$} (y0)
(y0) edge node [auto] {$\mu m$} (x3)
(c1) edge node [auto,yshift=-0.3em] {$M\eta$} (y0)
;
\node[below of=a0,xshift=1.25em,yshift=0em]{\ding{192}};
\node[below of=a0,xshift=7em,yshift=1em]{\ding{193}};
\node[below of=a0,xshift=18em,yshift=2.5em]{\ding{194}};
\node[below of=a0,xshift=24em,yshift=-1.5em]{\ding{195}};
\node[below of=a0,xshift=11em,yshift=-5.5em]{\ding{196}};
\node[below of=a0,xshift=19.7em,yshift=-6.9em]{\ding{197}};
\node[below of=a0,xshift=26em,yshift=-5.5em]{\ding{198}};
\node[below of=a0,xshift=18em,yshift=-10em]{\ding{199}};
\end{tikzpicture}
\end{equation*}
\ding{192} monad laws,
\ding{193} and
\ding{194} naturality of $\tau$,
\ding{195} $\app$ is the counit of the adjunction, and the definition of $s$
\ding{196}~naturality of $\tau$,
\ding{197} properties of strength,
\ding{198} monad laws,
\ding{199} coherence for $\tuple{MA,\mu,m,u}$.

Next, we need to  $\mladj s$ is a morphism between the monoid parts. It is easy to verify that it preserves the unit. As for the preservation of monoid multiplication, we first show that the following diagram commutes:
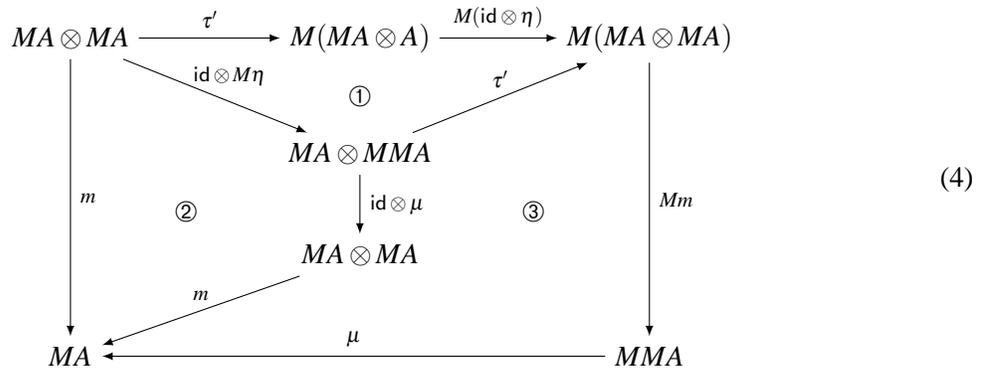
\begin{equation}\label{eq:start4}
  \begin{tikzpicture}
\node(a0){$MA \otimes MA$};
\node(a1)[below of=a0, yshift=-8em,yshift=1em]{$MA$};
\node(b0)[right of=a0, xshift=6em]{$M(MA \otimes A)$};
\node(b1)[below of=b0, xshift=0em]{$MA \otimes MMA$};
\node(b2)[below of=b1, xshift=0em,yshift=0.5em]{$MA \otimes MA$};
\node(c0)[right of=b0, xshift=6em]{$M(MA \otimes MA)$};
\node(c1)[below of=c0, yshift=-8em,yshift=1em]{$MMA$};
\path[arrows={-latex}, font=\scriptsize]
(a0) edge node [auto] {$\tau'$} (b0)
(b0) edge node [auto] {$M(\id \otimes \eta)$} (c0)
(a0) edge node [above,xshift=0.5em] {$\id \otimes M\eta$} (b1)
(b1) edge node [above] {$\tau'$} (c0)
(b1) edge node [auto] {$\id \otimes \mu$} (b2)
(b2) edge node [above] {$m$} (a1)
(a0) edge node [auto] {$m$} (a1)
(c0) edge node [auto] {$Mm$} (c1)
(c1) edge node [above] {$\mu$} (a1)
;
\node[below of=b0,yshift=2em]{\ding{192}};
\node[below of=b0,yshift=-2em,xshift=-6em]{\ding{193}};
\node[below of=b0,yshift=-2em,xshift=6em]{\ding{194}};
  \end{tikzpicture}
\end{equation}
\ding{192} naturality of $\tau'$,
\ding{193} monad laws,
\ding{194} the `twisted' coherence condition for $\tau'$ (see Remark~\ref{rem:symem}).

Now, to see that $\mladj s$ preserves the monoid multiplication, it is enough to show that the following diagram commutes, in which $E$ stands for $A \Rightarrow MA$:

\begin{equation*}
\begin{tikzpicture}
\node(a0){$(MA \otimes MA) \otimes A$};
\node(a1)[below of=a0, yshift=-38em]{$MA \otimes A$};
\node(a2)[right of=a1, xshift=4em]{$MA \otimes MA$};
\node(d0)[right of=a0, xshift=26em]{$(E \otimes E) \otimes A$};
\node(d1)[below of=d0, yshift=-6em]{$E \otimes (E \otimes A)$};
\node(d2)[below of=d1, yshift=-4em]{$E \otimes MA$};
\node(d3)[below of=d2, yshift=-4em]{$M(E \otimes A)$};
\node(d4)[below of=d3, yshift=-2em]{$MMA$};
\node(d5)[below of=d4, yshift=-6em]{$MA$};
\node(b1)[below of=a0, xshift=8em]{$(MA \otimes MA) \otimes MA$};
\node(b2)[below of=b1]{$MA \otimes (MA \otimes MA)$};
\node(b3)[below of=b2]{$MA \otimes M(MA \otimes A)$};
\node(b4)[below of=b3]{$MA \otimes M(MA \otimes MA)$};
\node(b5)[below of=b4]{$M(MA \otimes (MA \otimes MA))$};
\node(b6)[below of=b5]{$M(MA \otimes MA)$};
\node(b7)[below of=b6]{$MM(MA \otimes A)$};
\node(b8)[below of=b7]{$MM(MA \otimes MA)$};
\node(b9)[below of=b8, xshift=0em]{$MMMA$};
\node(c1)[right of=b1, xshift=8em]{$(E \otimes E) \otimes MA$};
\node(c2)[below of=c1]{$E \otimes (E \otimes MA)$};
\node(c3)[below of=c2]{$E \otimes M(E \otimes A)$};
\node(c4)[below of=c3]{$E \otimes MMA$};
\node(c5)[below of=c4]{$M(E \otimes MA)$};
\node(c6)[below of=c5]{$MM(E \otimes A)$};
\node(c7)[below of=c6]{$M(MA \otimes A)$};
\node(c8)[below of=c7]{$M(MA \otimes MA)$};
\node(c9)[below of=c8, xshift=0em]{$MMA$};
\path[arrows={-latex}, font=\scriptsize]
(a0) edge node [auto] {$(\mladj s \otimes \mladj s) \otimes \id$} (d0)
(a0) edge node [above, xshift=0.5em] {$\id \otimes \eta$} (b1)
(b1) edge node [auto] {$(\mladj s \otimes \mladj s) \otimes \id$} (c1)
(b2) edge node [auto] {$\mladj s \otimes (\mladj s \otimes \id)$} (c2)
(d0) edge node [above, xshift=-0.5em] {$\id \otimes \eta$} (c1)
(b1) edge node [auto] {$\cong$} (b2)
(c1) edge node [auto] {$\cong$} (c2)
(c2) edge node [auto] {$\id \otimes \tau'$} (c3)
(d0) edge node [auto] {$\cong$} (d1)
(d1) edge node [above,xshift=0.5em] {$\id \otimes (\id \otimes \eta)$} (c2)
(d1) edge node [above] {$\id \otimes \eta$} (c3)
(c3) edge node [auto] {$\id \otimes M\app$} (c4)
(d1) edge node [auto] {$\id \otimes \app$} (d2)
(d2) edge node [above] {$\id \otimes \eta$} (c4)
(c4) edge node [auto] {$\tau'$} (c5)
(c5) edge node [auto] {$M\tau'$} (c6)
(d2) edge node [above] {$\eta$} (c5)
(d2) edge node [auto] {$\tau'$} (d3)
(c6) edge node [auto] {$\mu$} (d3)
(c7) edge node [auto,xshift=1em] {$M(\mladj s \otimes \id))$} (d3)
(b2) edge node [auto] {$\id\otimes\tau'$} (b3)
(b3) edge node [auto,yshift=0.3em] {$\mladj s \otimes M(\mladj s \otimes \id)$} (c3)
(b3) edge node [auto] {$\id\otimes M(\id \otimes \eta)$} (b4)
(b4) edge node [auto] {$\mladj s \otimes Mm$} (c4)
(b4) edge node [auto] {$\tau'$} (b5)
(b5) edge node [auto] {$M(\mladj s \otimes m)$} (c5)
(b5) edge node [auto] {$M(\id \otimes m)$} (b6)
(b6) edge node [auto] {$M\tau'$} (b7)
(b7) edge node [auto] {$\mu$} (c7)
(c7) edge node [auto] {$M(\id \otimes \eta)$} (c8)
(b7) edge node [auto,,yshift=0.3em,xshift=2em] {$MM(\mladj s \otimes \id)$} (c6)
(b6) edge node [below,xshift=1em] {$M(\mladj s \otimes \id)$} (c5)
(c8) edge node [auto] {$Mm$} (d4)
(d3) edge node [auto] {$M\app$} (d4)
(d4) edge node [auto] {$\mu$} (d5)
(a0) edge node [auto] {$m \otimes \id$} (a1)
(a1) edge node [auto] {$\id \otimes \eta$} (a2)
(a2) edge node [auto] {$m$} (d5)
(b7) edge node [auto] {$MM(\id \otimes \eta)$} (b8)
(b8) edge node [auto] {$MMm$} (b9)
(b8) edge node [auto] {$\mu$} (c8)
(c9) edge node [auto] {$\mu$} (d5)
(b9) edge node [auto,yshift=-0.2em] {$\mu$} (d4)
(b9) edge node [auto,xshift=2em,yshift=-0.1em] {$M\mu$} (c9)
;
\draw [-latex, rounded corners=1em] (b6.west) -- ++(-2em,0) node[above, xshift=8em, yshift=-15em]{\scriptsize $Mm$} -- ++(0,-15em) -| (c9.south)
;
\node[below of=a0,xshift=2em,yshift=0em]{\ding{192}};
\node[right of=a0,xshift=10em,yshift=-1.5em]{\ding{193}};
\node[right of=a0,xshift=10em,yshift=-5.5em]{\ding{194}};
\node[right of=a0,xshift=21em,yshift=-5.5em]{\ding{195}};
\node[right of=a0,xshift=10em,yshift=-9.5em]{\ding{196}};
\node[right of=a0,xshift=19.2em,yshift=-10em]{\ding{197}};
\node[right of=a0,xshift=10em,yshift=-13.5em]{\ding{198}};
\node[right of=a0,xshift=21em,yshift=-13.5em]{\ding{199}};
\node[right of=a0,xshift=10em,yshift=-17.5em]{\ding{200}};
\node[right of=a0,xshift=19.2em,yshift=-18em]{\ding{201}};
\node[right of=a0,xshift=10em,yshift=-21em]{\ding{202}};
\node[right of=a0,xshift=10em,yshift=-24em]{\ding{203}};
\node[right of=a0,xshift=21em,yshift=-22em]{\ding{204}};
\node[right of=a0,xshift=16em,yshift=-26em]{\ding{205}};
\node[right of=a0,xshift=10em,yshift=-29.5em]{\ding{206}};
\node[right of=a0,xshift=21em,yshift=-29.5em]{\ding{207}};
\node[right of=a0,xshift=10em,yshift=-33.5em]{\ding{208}};
\node[right of=a0,xshift=21em,yshift=-36em]{\ding{209}};
\node[right of=a0,xshift=10em,yshift=-37.5em]{\ding{210}};
\end{tikzpicture}
\end{equation*}
\ding{192} see below,
\ding{193} $\otimes$ is a bifunctor,
\ding{194} and
\ding{195} naturality of $\cong$,
\ding{196}~naturality of $\tau'$,
\ding{197}~properties of strength,
\ding{198}~$\app$ is the counit of the adjunction,
\ding{199}~naturality of $\eta$,
\ding{200}~naturality of $\tau'$,
\ding{201}~properties of strength,
\ding{202}~$\otimes$~is a bifunctor,
\ding{203}~naturality of $\tau'$,
\ding{204}~naturality of $\eta$, and monad laws,
\ding{205}~and
\ding{206}~naturality of $\mu$, 
\ding{207}~$\app$ is the counit of the adjunction,
\ding{208}~naturality of $\mu$, 
\ding{209}~monad laws, 
\ding{210}~diagram~\eqref{eq:start4}.

Below, we detail \ding{192} from the diagram above.
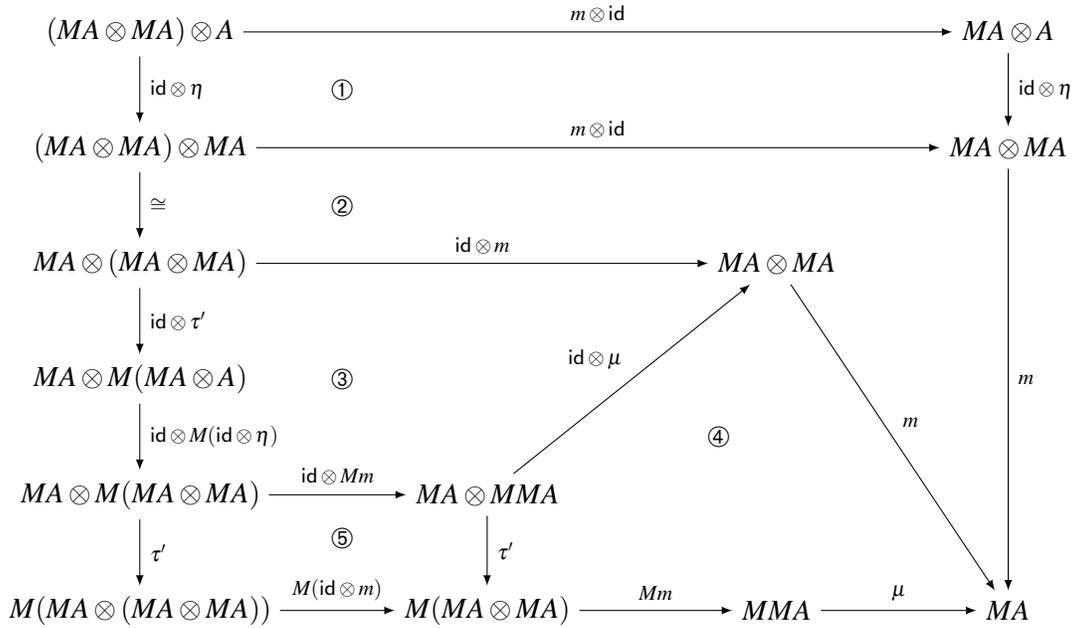
\begin{equation*}
  \begin{tikzpicture}
\node(a0){$(MA \otimes MA) \otimes A$};
\node(b1)[below of=a0, xshift=0em]{$(MA \otimes MA) \otimes MA$};
\node(b2)[below of=b1]{$MA \otimes (MA \otimes MA)$};
\node(b3)[below of=b2]{$MA \otimes M(MA \otimes A)$};
\node(b4)[below of=b3]{$MA \otimes M(MA \otimes MA)$};
\node(b5)[below of=b4]{$M(MA \otimes (MA \otimes MA))$};
\node(b6)[right of=b5,xshift=8em]{$M(MA \otimes MA)$};  
\node(b7)[right of=b6,xshift=6em]{$MMA$};
\node(b8)[right of=b7,xshift=4em]{$MA$};
\node(c0)[right of=a0, xshift=26em]{$MA \otimes A$};
\node(c1)[below of=c0]{$MA \otimes MA$};
\node(x2)[above of=b7, yshift=8em]{$MA \otimes MA$};
\node(x1)[above of=b6, yshift=0em]{$MA \otimes MMA$};
\path[arrows={-latex}, font=\scriptsize]
(a0) edge node [auto] {$\id \otimes \eta$} (b1)
(b1) edge node [auto] {$\cong$} (b2)
(b2) edge node [auto] {$\id\otimes\tau'$} (b3)
(b3) edge node [auto] {$\id\otimes M(\id \otimes \eta)$} (b4)
(b4) edge node [auto] {$\tau'$} (b5)
(b5) edge node [auto] {$M(\id \otimes m)$} (b6)
(b6) edge node [auto] {$Mm$} (b7)
(b7) edge node [auto] {$\mu$} (b8)
(a0) edge node [auto] {$m \otimes \id$} (c0)
(b1) edge node [auto] {$m \otimes \id$} (c1)
(c0) edge node [auto] {$\id \otimes \eta$} (c1)
(c1) edge node [auto] {$m$} (b8)
(b2) edge node [auto] {$\id\otimes m$} (x2)
(x2) edge node [auto] {$m$} (b8)
(b4) edge node [auto] {$\id \otimes Mm$} (x1)
(x1) edge node [auto] {$\tau'$} (b6)
(x1) edge node [auto] {$\id \otimes \mu$} (x2)
;
\node[below of=a0,xshift=7em,yshift=2em]{\ding{192}};
\node[below of=a0,xshift=7em,yshift=-2em]{\ding{193}};
\node[below of=a0,xshift=7em,yshift=-8em]{\ding{194}};
\node[below of=a0,xshift=20em,yshift=-10em]{\ding{195}};
\node[below of=a0,xshift=7em,yshift=-13.5em]{\ding{196}};
  \end{tikzpicture}
\end{equation*}
\ding{192} $\otimes$ is a bifunctor,
\ding{193} associativity of monoid multiplication,
\ding{194} $(\id \otimes \pholder$)-image of the diagram~\eqref{eq:start4},
\ding{195}~the `twisted' coherence condition for $\tau'$ (see Remark~\ref{rem:symem}),
\ding{196} naturality of $\tau'$.

It is left to show that $\mladj s$ is a split mono. The retraction $r$ is given as follows:
\begin{equation*}
r = \Bigl(
(A \Rightarrow MA) \xrightarrow{\cong}
(A \Rightarrow MA) \otimes I \xrightarrow{\id \otimes u}
(A \Rightarrow MA) \otimes MA \xrightarrow{\tau'}
M((A \Rightarrow MA) \otimes A) \xrightarrow{M \app}
MMA \xrightarrow{\mu}
MA
\Bigr)
\end{equation*}
We verify that $r \cdot \mladj s = \id$:
\begin{equation*}
\begin{tikzpicture}
\node(a0){$MA$};
\node(a1)[below of=a0]{$A \Rightarrow MA$};
\node(a2)[below of=a1]{$(A \Rightarrow MA) \otimes I$};
\node(a3)[below of=a2]{$(A \Rightarrow MA) \otimes MA$};
\node(a4)[below of=a3]{$M((A \Rightarrow MA) \otimes A)$};
\node(a5)[below of=a4]{$MMA$};
\node(a6)[below of=a5]{$MA$};
\node(b0)[right of=a1, xshift=8em]{$MA\otimes I$};
\node(b1)[below of=b0, yshift=-4em]{$MA \otimes MA$};
\node(b2)[below of=b1]{$M(MA \otimes A)$};
\node(b3)[below of=b2]{$M(MA \otimes MA)$};
\node(c0)[right of=b2, xshift=4em, yshift=-4em]{$MA\otimes MMA$};
\node(c1)[below of=c0, yshift=0em, xshift=10em]{$MA\otimes MA$};
\path[arrows={-latex}, font=\scriptsize]
(a0) edge node [auto] {$\mladj s$} (a1)
(a1) edge node [auto] {$\cong$} (a2)
(a2) edge node [auto] {$\id \otimes u$} (a3)
(a3) edge node [auto] {$\tau'$} (a4)
(a4) edge node [auto] {$M\app$} (a5)
(a5) edge node [auto] {$\mu$} (a6)
(b0) edge node [auto] {$\id \otimes u$} (b1)
(b1) edge node [auto] {$\tau'$} (b2)
(b2) edge node [auto] {$M(\id \otimes \eta)$} (b3)
(a0) edge node [auto] {$\cong$} (b0)
(b0) edge node [above, xshift=-0.5em] {$\mladj s \otimes \id$} (a2)
(b1) edge node [above] {$\mladj s \otimes\id$} (a3)
(b2) edge node [above] {$M(\mladj s \otimes\id)$} (a4)
(b3) edge node [above] {$Mm$} (a5)
(c0) edge node [auto] {$\id \otimes \mu$} (c1)
(b1) edge node [auto] {$\id \otimes M\eta$} (c0)
(c1) edge node [above] {$m$} (a6)
(c0) edge node [above] {$\tau'$} (b3)
(b1) edge[bend left=30] node [above, xshift=0.2em] {$\id$} (c1)
;
\node[right of=a0,yshift=-3em]{\ding{192}};
\node[right of=a0,xshift=3em,yshift=-8em]{\ding{193}};
\node[right of=a0,xshift=3em,yshift=-13.5em]{\ding{194}};
\node[right of=a0,xshift=3em,yshift=-17.5em]{\ding{195}};
\node[right of=a0,xshift=12.5em,yshift=-18em]{\ding{196}};
\node[right of=a0,xshift=20em,yshift=-18em]{\ding{197}};
\node[right of=a0,xshift=8em,yshift=-22em]{\ding{198}};
\end{tikzpicture}
\end{equation*}
\ding{192} naturality of $\cong$,
\ding{193} $\otimes$ is a bifunctor,
\ding{194} naturality of $\tau'$,
\ding{195} $\app$ is the counit of the adjunction, and the definition of $s$,
\ding{196} naturality of $\tau'$,
\ding{197} monad laws,
\ding{198} the `twisted' coherence condition for $\tau'$ (see Remark~\ref{rem:symem})

\end{document}